\documentclass[11pt,letterpaper]{article}
\usepackage{graphicx} 
\usepackage[utf8]{inputenc}
\usepackage[margin=1in]{geometry}
\usepackage{amsmath,amsthm,amssymb}
\usepackage{bbm}
\usepackage{thm-restate}
\usepackage{xcolor}
\usepackage{algorithm, algpseudocode}
\usepackage[pagebackref,hypertexnames=false]{hyperref}
\usepackage{float}
\usepackage[font=small]{caption}
\usepackage{subcaption}
\usepackage[authoryear]{natbib}
\usepackage{bbold}
\usepackage{mathtools}
\usepackage{doi}
\usepackage{tikz}
\usepackage{thm-restate}
\usepackage{xspace}
\usepackage{tcolorbox}
\usepackage[normalem]{ulem}
\usepackage{wrapfig}
\usepackage{booktabs}
\usepackage{makecell}
\usepackage{multirow}

\usepackage[nameinlink,capitalise]{cleveref}
\usepackage{crossreftools}

\newtoggle{COMMENTS}
\toggletrue{COMMENTS}

\usepackage{etoolbox} 
\newbool{inappendix}
\setbool{inappendix}{false}

\definecolor{DarkGreen}{rgb}{0.1,0.5,0.1}
\definecolor{DarkRed}{rgb}{0.5,0.1,0.1}
\definecolor{DarkBlue}{rgb}{0.1,0.1,0.5}
\definecolor{Gray}{rgb}{0.2,0.2,0.2}

\hypersetup{
    unicode=false,          
    pdftoolbar=true,        
    pdfmenubar=true,        
    pdffitwindow=false,      
    pdfnewwindow=true,      
    colorlinks=true,       
    linkcolor=DarkBlue,          
    citecolor=DarkGreen,        
    filecolor=DarkRed,      
    urlcolor=DarkBlue,          
    pdftitle={},
    pdfauthor={},
}

\makeatletter
\def\maketag@@@#1{\hbox{\m@th\normalfont\normalsize#1}}
\makeatother

\crefname{ALC@line}{line}{lines}    
\Crefname{ALC@line}{Line}{Lines}    
\crefname{assumption}{Assump.}{Assumps.}
\crefname{theorem}{Thm.}{Thms.}
\crefname{proposition}{Prop.}{Props.}
\crefname{lemma}{Lem.}{Lems.}
\crefname{equation}{Eq.}{Eqs.}
\crefname{section}{Sec.}{Secs.}
\crefname{appendix}{App.}{Apps.}
\crefname{table}{Tab.}{Tabs.}
\crefname{figure}{Fig.}{Figs.}
\crefname{example}{Ex.}{Exs.}
\crefname{corollary}{Cor.}{Cors.}
\crefname{algorithm}{Alg.}{Algs.}
\crefname{definition}{Def.}{Defs.}

\allowdisplaybreaks

\title{The Limits of Interval-Regulated Price Discrimination}

\newcommand*\samethanks[1][\value{footnote}]{\footnotemark[#1]}

\author{Kamesh Munagala\thanks{Department of Computer Science, Duke University, Durham, NC 27708-0129. Emails: \texttt{kamesh@cs.duke.edu}, \texttt{yiheng.shen@duke.edu}. Supported by NSF award IIS-2402823.} \and Yiheng Shen\samethanks[1] \and Renzhe Xu\thanks{Key Laboratory of Interdisciplinary Research of Computation and Economics, Shanghai University of Finance and Economics, China. Email: \texttt{xurenzhe@sufe.edu.cn}. Supported by the National Natural Science Foundation of China (No. 72442024) and the Shanghai Sailing Program (No. 24YF2711600). Part of this work was done while the author was visiting Duke University.}}
\date{}

\newtheorem{theorem}{Theorem}[section]
\newtheorem{lemma}[theorem]{Lemma}

\newtheorem{corollary}[theorem]{Corollary}
\newtheorem{proposition}[theorem]{Proposition}
\theoremstyle{definition}
\newtheorem{definition}{Definition}[section]
\newtheorem{example}{Example}[section]
\newtheorem{property}{Property}[section]
\theoremstyle{remark}

\newcommand{\R}{\mathbb{R}}
\newcommand{\X}{\mathcal{X}}

\newcommand{\A}{\mathcal{A}}

\newcommand{\Z}{\mathcal{Z}}
\newcommand{\bbI}{\mathbbm{1}}

\newcommand{\ba}{\mathbf{a}}
\newcommand{\bb}{\mathbf{b}}

\newcommand{\bx}{\mathbf{x}}

\newcommand{\by}{\mathbf{y}}

\newcommand{\tF}{\tilde{F}}
\newcommand{\tB}{\tilde{B}}

\newcommand{\tbx}{\tilde{\bx}}
\newcommand{\tx}{\tilde{x}}

\newcommand{\tX}{\tilde{\X}}

\DeclareMathOperator*{\argmax}{arg\,max}

\newcommand{\CS}{\mathrm{CS}}
\newcommand{\PS}{\mathrm{PS}}
\newcommand{\SW}{\mathrm{SW}}
\newcommand{\cs}{\mathrm{cs}}
\newcommand{\ps}{\mathrm{ps}}
\newcommand{\sw}{\mathrm{sw}}

\newcommand{\remain}{\mathrm{remain}}

\newcommand{\uni}{\mathrm{uniform}}
\newcommand{\ui}{{\underline{i}}}

\newcommand{\fopt}{\mathrm{OptPrice}^{F}}
\newcommand{\ZAP}{\Z_{\mathrm{AP}}}
\newcommand{\ZAC}{\Z_{\mathrm{AC}}}

\newcommand{\ZPP}{\Z_{\mathrm{PP}}}

\newcommand{\ZPC}{\Z_{\mathrm{PC}}}

\newcommand{\ZPMin}{\Z_\mathrm{PMin}}
\newcommand{\ZAMin}{\Z_\mathrm{AMin}}

\newcommand{\uniform}{\mathrm{uniform}}
\DeclareMathOperator{\support}{\mathrm{Supp}}

\DeclareMathOperator{\optPrice}{\mathrm{OptPrice}}
\DeclareMathOperator{\ERM}{\mathrm{ERM}}
\DeclareMathOperator{\mass}{\mathrm{mass}}
\DeclareMathOperator{\equal}{\mathrm{eq}}

\newcommand{\standardization}{\mathrm{Standardize}}

\newcommand{\scheme}{market scheme\xspace}
\newcommand{\schemes}{market schemes\xspace}

\newcommand{\SCHEME}{Market Scheme\xspace}
\newcommand{\SCHEMES}{Market Schemes\xspace}
\newcommand{\Fvalid}{$F$-valid\xspace}
\newcommand{\standard}{standard-form and $F$-valid\xspace}

\newcommand{\ActiveCSMin}{\CS_A^{\min}}

\newcommand{\ActivePSMin}{\PS_A^{\min}}

\newcommand{\PassiveCSMin}{\CS_P^{\min}}

\newcommand{\PassiveSWMax}{\SW^{\max}}

\newcommand{\UniRev}{R_{\uniform}}

\begin{document}

\maketitle
\begin{abstract}
In this paper, we study third-degree price discrimination in a model first presented by \citet*{bergemann2015limits}. Since such price discrimination might create market segments with vastly different posted prices, we consider regulating these prices, specifically, by restricting them to lie within an interval. Given a price interval, we consider segmentations of the market where a seller, who is oblivious to the existence of such regulation, still posts prices within the price interval.

We show the following surprising result: For any market and price interval where such segmentation is feasible, there is always a different segmentation that optimally transfers all excess surplus to the consumers. In addition, we characterize the entire space of buyer and seller surplus that is achievable by such segmentation, including maximizing seller surplus, and simultaneously minimizing buyer and seller surplus.

A key technical challenge is that the classical segmentation method of \citet*{bergemann2015limits} fails under price constraints. To address this, we develop three intuitive but fundamentally distinct segmentation constructions, each tailored to a different surplus objective. These constructions maintain different invariants, reflect different economic intuitions, and collectively form the core of our regulated surplus characterization.

\end{abstract}

\clearpage

\makeatletter
\setcounter{tocdepth}{2}
\makeatother
\tableofcontents

\thispagestyle{empty}
\clearpage

\section{Introduction}

We investigate third-degree price discrimination in markets characterized by a seller with an unlimited supply of goods who aims to implement differential pricing strategies across various buyer sub-populations. Drawing on the seminal work by \citet{bergemann2015limits}, we focus on information intermediaries in these market structures. With access to rich consumer data and powerful machine learning tools, intermediaries can accurately infer buyer valuations and finely segment the buyer population. This segmentation allows the seller to tailor their pricing strategies to each sub-population. Such scenarios are especially common in modern two-sided e-commerce platforms—most notably in ad exchanges—where the platform acts as an intermediary between ad slot publishers (sellers) and advertisers (buyers), enabling optimized price discrimination through market segmentation.

In this context, consumer and producer benefits are quantified as consumer surplus and producer surplus, respectively. Consider a scenario where a buyer values the good at $v$ and is charged a price $p$. The purchase decision hinges on the condition $v \ge p$, with the consumer's benefit being $v - p$ and the seller's benefit being $p$. Consumer surplus and producer surplus are defined as the aggregate benefits accruing to buyers and sellers, respectively. Now, consider a market characterized by a set of buyers' valuations, $V = \{v_1, v_2, \dots, v_n\}$. The structure of the market is captured by the vector $\bx^* = (x^*_1, x^*_2, \dots, x^*_n) \in \R^n_{\ge 0}$, where $x^*_i$ denotes the number of buyers with a valuation of $v_i$. Without segmentation, the seller is constrained to applying uniform pricing, charging a single price $p$ across the market. Under these conditions, the consumer surplus is calculated as $\sum_{i=1}^n \bbI[v_i \ge p]\cdot x^*_i(v_i-p)$, and the producer surplus as $\sum_{i=1}^n \bbI[v_i \ge p] \cdot x^*_ip$. The seller, aiming to maximize revenue—synonymous with producer surplus—will select the optimal uniform price $p^*$.

\paragraph{Information Intermediary.}
An intermediary with full knowledge of consumer valuations can affect consumer and producer benefits by segmenting the market $\bx^*$ into $Q$ segments $\{\bx_1, \bx_2, \dots, \bx_Q\}$, where each $\bx_q \in \R^n_{\ge 0}$ and $\sum_{q=1}^{Q} \bx_q = \bx^*$, ensuring coverage of the whole market.
The seller is then permitted to set a single price $p_q$ for each market segment $\bx_q$.
Intuitively, market segmentation in third-degree price discrimination should benefit sellers, as it allows them to charge optimal prices tailored to different sub-populations, thereby increasing revenue.
However, the seminal work of \citet{bergemann2015limits} established a surprising ``consumer optimality'' result: an intermediary can segment the market such that the seller's revenue (or producer surplus) remains unchanged from uniform pricing, while all excess surplus is allocated to buyers.
Therefore, the sum of consumer surplus and producer surplus is equal to the social welfare, $\SW^{\max} = \sum_{i=1}^n x_i^*v_i$. In other words, the item always sells, while the seller's revenue is exactly that from uniform pricing! 

\begin{example} \label{ex:intro_four}
    Suppose the buyers have four possible values $V = \{1, 2, 3, 6\}$, and the aggregate market is given by $\bx^* = (0.36, 0.20, 0.18, 0.26)$. The optimal uniform price $p^*$ is $6$, yielding a producer surplus of $1.56$. However, the consumer surplus in this case is zero: buyers with valuations in $\{1, 2, 3\}$ do not purchase the good, and those with valuation $6$ obtain no surplus when the price equals their value. Now consider segmenting the market into four parts, with value distributions $(0.36, 0.12, 0.12, 0.12)$, $(0.00, 0.08, 0.06, 0.07)$, and $(0.00, 0.00, 0.00, 0.07)$, and associated prices $1$, $2$, and $6$, respectively. It is straightforward to verify that each price is optimal for its respective segment. The resulting producer surplus remains $1.56$, matching that under uniform pricing, while the consumer surplus increases to $1.30$. Moreover, the good is always sold in each segment, since the posted price exactly matches the lowest valuation among the buyers in that segment. This ensures maximal consumer surplus.
\end{example}

In addition, \citet{bergemann2015limits} characterize the full region of all possible consumer-producer surplus pairs achievable through different market segmentations, illustrated as the gray region in \cref{fig:summary}. Notably, each such pair corresponds to a feasible solution of a linear program~\citep{CummingsD0W20} and is thus computable in polynomial time. However, the linear program itself provides limited insight into whether consumer-optimal outcomes can be achieved or the structural properties of the full surplus region.

\paragraph{Regulation via Price Intervals.}
The seminal result of \citet{bergemann2015limits} highlights the striking flexibility of surplus division achievable through market segmentation by a fully informed intermediary. However, their model assumes that the intermediary can induce arbitrary segmentations without constraints, often resulting in highly heterogeneous pricing across segments. While theoretically appealing, such unconstrained pricing can be problematic in practice. From a platform’s perspective, large price disparities across consumer segments may lead to buyer envy and erode trust, as users perceive the pricing as arbitrary or unfair~\citep{haws2006dynamic,alderighi2022consumer}. From a broader societal standpoint, excessively high prices may exclude buyers and reduce market activity~\citep{mankiw1998principles}, while extremely low prices can undermine brand value~\citep{inderst2024price} or exert pressure on small businesses~\citep{besanko2014economics}.

To mitigate these concerns, many real-world regulatory bodies impose pricing constraints through price intervals. For example, EU Regulation 2018/1971 sets upper bounds on intra-EU communication charges. In the United States, several states—including California and New York—enforce price-gouging limits during emergencies, as codified in California Penal Code~\S396 and New York General Business Law~\S396-R. These constraints help reduce price dispersion, enhance consumer perceptions of fairness, and promote overall market stability.

In light of these practices, we study the effect of imposing interval constraints on price discrimination. Specifically, we model regulation by requiring that all prices lie within a contiguous subset $F = \{v_\ell, v_{\ell+1}, \dots, v_r\} \subseteq V$ of the value space. 
This modeling choice offers a tractable way to analyze regulatory impacts while capturing common policy instruments that limit price variability.


Our model admits both a design and a robust predictions interpretation. From a design perspective, it captures an intermediary who segments the market while deliberately limiting price variation—motivated by fairness, reputational concerns, or regulatory pressure. The interval constraint reflects a preference to avoid large disparities across segments. Alternatively, the model can be seen as a robust predictions exercise: in markets with price discrimination but relatively narrow observed price ranges~\citep{anania2014price,duch2021online}, what surplus divisions remain feasible? This perspective is particularly relevant when the intermediary’s flexibility is limited. More broadly, by limiting price dispersion, the model serves as a stylized proxy for reducing surplus or price variance---objectives discusses in prior work~\citep{XuZ00SX22,CohenEL22,yang2024fairness,AliLV20}.

We analyze two regimes reflecting different degrees of intermediary control over seller behavior. In the \emph{passive intermediary} setting, inspired by platforms such as local food delivery services, it is hard to enforce price regulations on sellers directly. However, the intermediary may influence seller behavior by recommending prices that lie in the seller’s optimal set. Thus, regulation is enforced indirectly by designing segmentations where the seller's revenue-maximizing price lies in the interval $F$. In contrast, the \emph{active intermediary} setting models platforms like Amazon or eBay, where the intermediary exerts direct control over pricing. In this case, the seller is restricted to choosing a price within $F$ for each segment, regardless of whether it maximizes revenue.




\subsection{Overview of Results}

\subsubsection{Passive Intermediary Model} \label{sec:passive_intro}
Our main results are for the more challenging passive intermediary model. Note that there could be price sets $F$ for which no segmentation is possible. We therefore focus on price sets $F$ where there is at least one feasible segmentation, and we call such price sets {\em feasible}. Now, like~\citet{bergemann2015limits}, we can ask: {\em Given feasible $F$, how much surplus can be transferred to the consumers via a feasible segmentation?} 

\paragraph{Feasibility with the Optimal Uniform Price Excluded.} We begin by examining the conditions under which a price set $F$ is feasible. Any set $F$ that includes an optimal uniform price is trivially feasible, as the intermediary can simply treat the whole market as one segment. However, this inclusion is not always necessary—the optimal uniform price(s) may be either strictly greater or strictly smaller than all values in $F$. We present an example for this scenario first.

\begin{example}[Continuation of \cref{ex:intro_four}] \label{ex:intro-four-F-outside}
    In the setting of \cref{ex:intro_four}, suppose the regulated price set is $F = \{2, 3\}$. The intermediary can segment the market into $(0.00, 0.20, 0.00, 0.09)$ and $(0.36, 0.00, 0.18, 0.17)$, with optimal prices $2$ and $3$, respectively—both within $F$. Thus, $F$ is feasible even though the uniform optimal price $p^* = 6$ lies outside $F$.
\end{example}



In \cref{sect:justification-feasibility}, we formally analyze cases where the regulated set excludes the optimal prices yet remains feasible. We first show that feasible regulated price sets can deviate significantly from the optimal uniform prices. We then derive sufficient conditions for feasibility and validate their prevalence through simulations under uniform distributions, showing that 20\% of sets excluding the optimal price remain feasible in over 60\% of cases.


A regulated set excluding the optimal uniform prices $p^*$ has significant practical implications. If $p^*$ exceeds all values in $F$, any segmentation that satisfies the regulation guarantees strictly higher consumer surplus than uniform pricing, as all buyers with valuations at least $p^*$ will purchase the good at a lower price within $F$. As shown in \cref{ex:intro_four}, with $F = \{2, 3\}$ and $p^* = 6$, buyers with valuation $6$ always purchase at $2$ or $3$, ensuring a positive surplus, whereas under uniform pricing, consumer surplus is zero. This contrasts with the unregulated setting in \citet{bergemann2015limits}, where minimal consumer surplus can be zero.  Furthermore, in \cref{app:final_discussion}, we explore how regulators can design $F$ to maximize worst-case consumer surplus, protecting consumers' benefits across all segmentations that comply with the regulation. Utilizing the techniques developed later, our result suggests designing $F$ with minimal endpoints $v_{\ell}$ and $v_r$, potentially excluding $p^*$.

\paragraph{Feasibility Implies Consumer Optimality.} 
For a feasible $F$, let $\min(F) = \min_{v \in F} v$. Clearly, if $v_i < \min(F)$, there is no segmentation where a buyer with value $v_i$ can buy the item. Therefore, the maximum attainable total surplus is $\SW^{\max}(\bx^*, F) = \sum_{i: v_i \ge \min(F)} x_i^*v_i$, where the item is always sold if the value of a buyer is at least $\min(F)$. Our main result is the following surprising consumer-optimality theorem, proved in \cref{sec:passive-inter}:

\begin{theorem}[Informal; see \cref{thrm:passive-overview}] \label{thm:informal_main} \label{thm:informal-main}
In the passive intermediary model, for any feasible $F$, there is a market segmentation where the seller revenue equals that of uniform pricing without regulation \footnote{The seller revenue in this model is always at least that of uniform pricing, since the seller can choose to ignore the segmentation.}, while the total surplus is precisely $\SW^{\max}(\bx^*, F)$, hence transferring all remaining surplus to the buyers.
\end{theorem}


Note that if $F$ includes the optimal uniform price $p^*$, the segmentation method of \citet{bergemann2015limits} can be directly applied to achieve the desired goal (see \cref{app:bbm}). While one might view \cref{thm:informal_main} as a straightforward extension of their result, this is not the case. Their approach crucially depends on $F$ containing $p^*$, which ensures that $p^*$ remains optimal when iteratively removing market segments. In our setting, however, this assumption may fail: $p^*$ can lie outside $F$, as shown in \cref{ex:intro-four-F-outside}, making their approach invalid. It is then unclear a priori why the theorem should still hold. Our key contribution is a novel segmentation scheme, based on a different set of invariants, that proves the result in this broader setting. The next example illustrates why the method of \citet{bergemann2015limits} fails here.

\begin{example}[Continuation of \cref{ex:intro_four}] \label{ex:intro_four_continue}
    If we directly apply the method of \citet{bergemann2015limits} to \cref{ex:intro_four} (see \cref{tab:construction-bbm}), one of the resulting market segments is $(0.00, 0.00, 0.00, 0.07)$, which has a unique optimal price of $6$—outside the regulated set $F$. Therefore, this segmentation violates the regulation constraints, as not all prices lie within $F$.
\end{example}


\paragraph{Producer-Consumer Surplus Region.} 
\begin{figure}[t]
    \centering
    \includegraphics[width=0.6\linewidth]{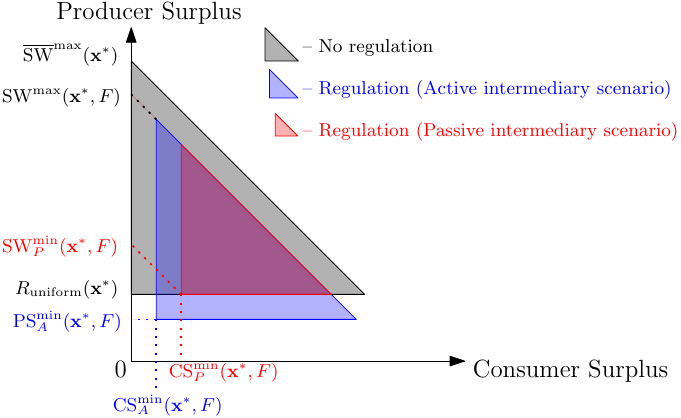}
    \caption{The illustration of the CS-PS feasible area in three different types of regulation. 
    }
    \label{fig:summary}
\end{figure}
So far, we have focused on transferring surplus to the buyers. A natural next step is to identify the seller-optimal market segmentation and characterize all possible ways to split the surplus between buyers and the seller.

We answer these questions in \cref{sec:passive-inter} by exactly characterizing the entire set of achievable producer surplus, consumer surplus pairs for any feasible $F$. We show that the region is a right triangle, as depicted in red in \cref{fig:summary}. The rightmost point is achieved by the buyer-optimal segmentation in \cref{thm:informal_main}. We also present non-trivial market segmentations for the other two extreme points of the triangle: the top-left point which achieves the maximum total surplus $\SW^{\max}(\bx^*, F)$ and the maximum producer surplus simultaneously, and the bottom-left point which achieves minima of both consumer and producer surplus simultaneously. Any point within the triangle is achieved by randomizing between these three schemes. Unlike \cref{thm:informal_main}, even if $F$ includes an optimal uniform price, \citet{bergemann2015limits}'s method does not help identifying the region, as shown in \cref{app:bbm}.

\paragraph{Insights.}  
Our characterization yields several key insights. First, consumer-optimality persists under regulation: for any feasible price set $F$, a segmentation exists that maintains the seller’s uniform-price revenue while allocating all remaining surplus to consumers. Second, the minimum consumer surplus is independent of seller strategy—it remains unchanged whether the seller maximizes or minimizes revenue. Third, the seller-optimal segmentation also maximizes total welfare.

From a regulatory standpoint, the triangular surplus region offers a clear view of the fairness-efficiency trade-off. Notably, the left edge of the region gives an exact lower bound on consumer surplus across all feasible segmentations, providing a robust benchmark for regulatory design. Further discussion is provided in \cref{app:final_discussion}.

\paragraph{Technical Contribution: Segmentation Algorithms.}
A central technical contribution of our work is the design of segmentation algorithms that attain the three extreme points of the regulated surplus region. While these constructions may appear conceptually similar to those in \citet{bergemann2015limits}, they differ fundamentally due to the presence of regulatory constraints, which render the original method inapplicable---even when the optimal uniform price lies within the regulated set $F$ (see \cref{app:bbm} for discussion).

Each extremal point requires a distinct segmentation strategy. The seller-optimal scheme, while still extracting equal-revenue segments from the residual market, carefully selects the support of each segment to maximize producer surplus. The consumer-optimal scheme strategically switches from a seller-favoring process to a buyer-favoring one, maximizing consumer surplus while preserving feasibility. The social-welfare-minimizing scheme further modifies the residual market to simultaneously reduce both consumer and producer surplus. These constructions maintain distinct invariants and cannot be unified into a single procedure. The design and analysis of these algorithms—ensuring feasibility under regulation while achieving each extremal surplus benchmark—are non-trivial and form the technical core of our paper.

\subsubsection{Active Intermediary Model}
Recall that in this model, the seller is required to set a price within the regulated set $F$ for each market segment—even if the seller's optimal price lies outside $F$. As a result, \emph{any} price set $F$ is feasible in the active intermediary model, since the seller is forced to comply with the regulation. This restriction simplifies the analysis compared to the passive intermediary model.

Note that the uniform pricing revenue under this constraint can be strictly lower than that obtained by selecting the globally optimal price from $V$. We define the constrained uniform pricing revenue as $
\ActivePSMin(\bx^*, F) = \max_{v \in F} \{ v \cdot \sum_{i: v_i \ge v} x^*_i \}$. As in the passive case, the maximum achievable total surplus remains $
\SW^{\max}(\bx^*, F) = \sum_{i: v_i \ge \min(F)} x^*_i \cdot v_i$, since all buyers with value at least $\min(F)$ will purchase the good under any $F$-instructed scheme.

Our main result for the active intermediary model parallels that of the passive case: the entire region of achievable consumer and producer surplus is characterized by three inequalities. The full model and proof are provided in \cref{sec:active-inter}.

\begin{theorem}[Informal, see \cref{thm:RSeller_base_bounds}]
In the active intermediary model, for any $F$, there is a market segmentation where the seller revenue equals that of uniform pricing within $F$ ($\ActivePSMin(F)$), while the total surplus is $\SW^{\max}(\bx^*, F)$, hence transferring all remaining surplus to the buyers.
\end{theorem}

Analogous to the passive intermediary model, we characterize the entire region of achievable producer and consumer surplus in the active intermediary setting and show that it again forms a right triangle—illustrated as the blue region in \cref{fig:summary}. To interpret this figure, note that any market scheme which passively induces the seller to post prices within $F$ will yield the same consumer and producer surplus under the active model. As a result, the feasible region for the passive intermediary is a subset of that for the active intermediary. Moreover, the hypotenuses of both right triangles lie on the same line, representing the same total social welfare: $\sum_{i: v_i \ge \min(F)} x^*_i v_i$. Interestingly, we observe that the consumer surplus achievable in the active intermediary model can exceed that in the unregulated case. This is because the seller can be \emph{forced} to post lower prices within $F$, even when his optimal price lies outside $F$, thereby transferring additional surplus to the buyers.

\subsection{Additional Related Work}

\paragraph{Information Design and Persuasion.} Our model of price discrimination falls in the general framework of information design and persuasion, first presented by \cite{kamenica2011bayesian}. In this framework, there is a sender (she) and a receiver (he). The receiver takes an action based on his perception of the state of nature that yields him optimal utility. However, the sender has a different utility function over actions of the receiver and the state of nature.  The sender knows more information about the state of nature than the receiver, and is able to signal her information to the receiver. The receiver updates his belief about the state of nature using this signal and subsequently acts optimally (under his own utility function). The sender's goal is to signal in a fashion that maximizes her own utility function.  Several problems in information design fall in this general framework, and hence it has been widely studied~\citep{rayo2010optimal,wang2013bayesian,kolotilin2017persuasion,kamenica2017competition,DBLP:conf/sigecom/DughmiKQ16,DBLP:journals/jet/ArieliB19,DBLP:journals/siamcomp/DughmiX21,DBLP:conf/innovations/BabichenkoTXZ22,DBLP:journals/geb/BabichenkoTXZ22,DBLP:journals/jet/MatyskovaM23,DBLP:conf/sigecom/KremerMP13,DBLP:journals/mktsci/ChakrabortyH14,bergemann2015limits,DBLP:conf/aaai/XuRDT15,DBLP:conf/atal/RabinovichJJX15,DBLP:journals/ior/MansourSS20,haghpanah2023pareto}. We refer readers to \cite{dughmi2017algorithmic,kamenica2019bayesian,bergemann2019information} for surveys on the literature.

\paragraph{Price Discrimination.} Third-degree price discrimination is a special case of persuasion where the state of nature is the buyer's valuation. The intermediary is the sender who knows the buyers' valuation and signals this to the seller. The seller is the receiver who sets a price (the action) to maximize his revenue given the signal. This model was first developed by~\citet{bergemann2015limits}. They showed optimal signaling schemes that preserve seller revenue while transferring the rest of the surplus to the buyers. \citet{CummingsD0W20} subsequently specified the set of all buyer-optimal signaling schemes by a linear program; this also follows from a more general result for persuasion in \citet{DBLP:journals/siamcomp/DughmiX21}. 

This general model has been extended to continuous distributions \citep{DBLP:conf/atal/ShenTZ18a}; auction design with multiple buyers \citep{DBLP:conf/sigecom/AlijaniBMW22}; a partially informed intermediary \citep{CummingsD0W20,arieli2024robust}; keyword search auctions \citet{DBLP:conf/www/BergemannDLZ22}; privacy constraints \citep{fallah2024limits}; a buyer with a public budget or a deadline \citep{DBLP:conf/sigecom/KoM22}; multiple rounds of interactive persuasion \citep{DBLP:conf/innovations/MaoLW22};  unifying with the second price discrimination \citep{bergemann2024unified}; and game-theoretic formulations of seller and buyer behavior \citep{AliLV20}. 
Our work on price regulation is a way to impose {\em fairness} on the discrimination process. The work of \cite{BanerjeeMSW24} also considers fairness and present signaling schemes that optimize concave functions of buyer surplus rather than only focusing on total surplus. In contrast, we achieve fairness via making the intermediary's signal restrict the seller's optimal actions (prices set for different signals), which leads to entirely different schemes and techniques. 

\paragraph{Price Regulation.}
Price discrimination has become a reality due to the rapidly increasing amount of data collected from consumers, which enables ML models to predict buyer valuations. Due to concerns about fairness and privacy, regulations against arbitrary price discrimination have become a heated topic in policy discussions \citep{BigDataWH,gee2018fair,gerlick2020ethical,gillis2020false}. Several recent works have proposed models and objectives for price regulation \citep{acquisti2015privacy,acquisti2016economics,cowan2018regulating,kallus2021fairness}, and studied the implications of price regulation on the social welfare and the surplus of producers (firms) and consumers (buyers) \citep{chen2023data,fallah2024limits,yang2024fairness, yang2023regulating,CohenEL22, XuZ00SX22}. 
These works focus on directly regulating the seller; however, our work {\em indirectly regulates} the seller via the intermediary's behavior.
To the best of our knowledge, we are the first to study price regulation in the price discrimination model of \citep{bergemann2015limits}, where the intermediary selectively reveals information to make an unconstrained seller's optimal behavior mimic regulation.

\subsection{Roadmap}
\cref{sec:setting} introduces the model and key definitions for the passive intermediary setting. In \cref{sec:passive-inter}, we present algorithms that construct the three market schemes achieving the extreme points of the consumer-producer surplus region. Full proofs are deferred to \cref{sect:proof-passive}. The active intermediary model is analyzed in \cref{sec:active-inter}, with additional discussion of the passive model provided in \cref{sect:discussion-passive-model}.

\section{Preliminaries for Passive Intermediary Model} \label{sec:setting}
\subsection{Information Intermediary Model}
We consider a setting where a seller offers a good to a finite population of buyers, each demanding at most one unit and having a private value drawn from a known finite set $V = \{v_1, v_2, \ldots, v_n\}$, with $v_1 < v_2 < \cdots < v_n$. A market is represented by a vector $\bx = (x_1, \ldots, x_n) \in \mathbb{R}_{\ge 0}^n$, where $x_i$ denotes the mass of buyers with value $v_i$. The total market mass is $\mass(\bx) = \sum_{i=1}^n x_i$. We denote the aggregate market by $\bx^*$ and, without loss of generality, normalize it so that $\mass(\bx^*) = 1$.

For any price $v \in V$, the cumulative demand is defined as $G_{\bx}(v) = \sum_{i: v_i \ge v} x_i$, representing the mass of buyers willing to purchase at price $v$ or higher. The support of market $\bx$ is $\support(\bx) = \{v_i \in V : x_i > 0\}$. For coordinate-wise comparisons, we write $\ba \le \bb$ if $a_i \le b_i$ for all $i \in [n]$.

If the seller posts a price $p \in V$, his revenue on market $\bx$ is $R_{\bx}(p) = p \cdot G_{\bx}(p)$. The set of revenue-maximizing prices is denoted by $\optPrice(\bx) = \arg\max_{v \in V} R_{\bx}(v)$. This set may contain multiple prices, and we assume the seller is indifferent among them.

The intermediary observes each buyer’s value and segments the market, recommending prices to the seller. The resulting structure is referred to as the \scheme.

\begin{definition}[\SCHEME]
    A \emph{\scheme} is a collection $\Z = \{(\bx_q, p_q)\}_{q \in [Q]}$, where each $\bx_q \in \mathbb{R}_{\ge 0}^n$ and $p_q \in V$. The collection $\{\bx_q\}_{q \in [Q]}$ forms a \emph{segmentation} of the aggregate market $\bx^*$, meaning $\sum_{q \in [Q]} \bx_q = \bx^*$. Each pair $(\bx_q, p_q)$ represents a market segment along with the price that the intermediary instructs the seller to post for that segment.
\end{definition}

We will later introduce two regulation models that differ in how strictly the seller adheres to these instructions.

Given a market segment and price pair $(\bx, p)$, the consumer surplus, producer surplus, and social welfare are defined as
\[
\cs(\bx, p) = \sum_{i: v_i \ge p} (v_i - p) \cdot x_i, \quad
\ps(\bx, p) = \sum_{i: v_i \ge p} p \cdot x_i, \quad
\sw(\bx, p) = \sum_{i: v_i \ge p} v_i \cdot x_i.
\]
These definitions extend naturally to a \scheme $\Z = \{(\bx_q, p_q)\}_{q \in [Q]}$, for which the aggregate consumer surplus, producer surplus, and social welfare are defined as:
\[
\CS(\Z) = \sum_{q \in [Q]} \cs(\bx_q, p_q), \quad
\PS(\Z) = \sum_{q \in [Q]} \ps(\bx_q, p_q), \quad
\SW(\Z) = \sum_{q \in [Q]} \sw(\bx_q, p_q).
\]
By construction, the identity $\CS(\Z) + \PS(\Z) = \SW(\Z)$ holds for any scheme $\Z$.

\subsection{Passive Intermediary Model} 
\label{sec:model_passive}

We now introduce regulation.\footnote{We defer definitions and results for the active intermediary model to \cref{sec:active-inter}. In the main body of the paper, we focus on the passive intermediary model.} Recall that the intermediary aims to ensure that the seller sets prices within a designated price range. A regulated price set $F$ is a subset of the values $V$. Throughout the paper, we assume that $F$ is a contiguous subset of $V$, i.e., $F = \{v_\ell, v_{\ell + 1}, \dots, v_r\}$ for some $\ell \le r$.

\begin{definition}[$F$-valid \SCHEME] \label{def:inter_IC}
    For a contiguous price set $F = \{v_\ell, v_{\ell+1}, \dots, v_r\} \subseteq V$, a \scheme $\Z = \{(\bx_q, p_q)\}_{q \in [Q]}$ is said to be \emph{$F$-valid} if for all $q \in [Q]$, the instructed price $p_q$ satisfies $p_q \in F \cap \optPrice(\bx_q)$.
\end{definition}

That is, a \scheme is $F$-valid if (i) every instructed price lies within the regulated set $F$, and (ii) the seller is willing to adopt the instructed price in each segment, since it is revenue-maximizing there. The seller thus has no incentive to deviate. As illustrated in \cref{ex:intro-four-F-outside}, the two market segments with the optimal prices together form an $F$-valid \scheme.

When restricting attention to $F$-valid schemes, such a scheme may not exist for every aggregate market and every price set $F$. This leads to the following feasibility condition.

\begin{definition}[$\bx^*$-Feasible Price Set]
    Given a market $\bx^*$, a contiguous price set $F \subseteq V$ is said to be \emph{$\bx^*$-feasible} if there exists an $F$-valid \scheme $\Z = \{(\bx_q, p_q)\}_{q \in [Q]}$ for $\bx^*$.
\end{definition}



This work seeks to answer the following fundamental question:

\begin{quote}
    \emph{Given a market $\bx^*$ and an $\bx^*$-feasible price set $F$, what is the achievable region of consumer surplus and producer surplus under $F$-valid \schemes?}
\end{quote}

\section{Market Schemes for the Passive Intermediary Model} \label{sec:passive-inter}
In this section, we address the question posed in \cref{sec:model_passive}. The formal results are stated in \cref{sect:result}, followed by a high-level outline of the proof in \cref{sect:outline-analysis}. Complete proofs are deferred to \cref{sect:proof-passive}.

\subsection{Main Theorem and Overview} \label{sect:result}

We begin by summarizing three key quantities that will anchor our characterization:
\begin{itemize}
    \item \textbf{Maximal social welfare}: The maximum social welfare achievable when instructed prices are restricted to the set $F = \{v_\ell, \ldots, v_r\}$ is upper bounded by the expected value of the buyer conditional on their value being at least $v_\ell$. Formally,
    \[
    \PassiveSWMax(\bx^*, F) = \sum_{i=\ell}^n x_i^* v_i.
    \]

    \item \textbf{Minimal producer surplus}: This is the seller’s optimal revenue in the absence of segmentation, i.e., the revenue from posting a single price:
    \[
    R_\uni(\bx^*) = \max_{p \in V} R_{\bx^*}(p).
    \]

    \item \textbf{Minimal consumer surplus}: Let $\A(\bx^*, F)$ denote the set of all $F$-valid \schemes{} for the market $\bx^*$. The minimal achievable consumer surplus over this set is defined as
    \[
    \PassiveCSMin(\bx^*, F) = \min_{\Z \in \A(\bx^*, F)} \CS(\Z).
    \]
    A detailed characterization of $\PassiveCSMin(\bx^*, F)$ is deferred to \cref{sect:proof-passive-SW}.
\end{itemize}

Our main result shows that these three quantities fully characterize the feasible region of achievable consumer and producer surplus under passive regulation.

\begin{theorem}[Main Theorem] \label{thrm:passive-overview}
    Let $\bx^*$ be a market and let $F = \{v_\ell, v_{\ell+1}, \dots, v_r\} \subseteq V$ be a contiguous, $\bx^*$-feasible price set. Then the set of achievable $(\text{consumer surplus}, \text{producer surplus})$ pairs is exactly the set of points $(x, y)$ satisfying: (1) $x + y \le \PassiveSWMax(\bx^*, F)$, (2) $x \ge \PassiveCSMin(\bx^*, F)$, and (3) $y \ge R_\uni(\bx^*)$.
\end{theorem}

The red triangle in \cref{fig:summary} illustrates the region described in \cref{thrm:passive-overview}. This region lies strictly within the full consumer-producer surplus frontier described by \citet{bergemann2015limits} (shown as the larger gray triangle), but the lower bound on producer surplus is the same in both.

The proof of \cref{thrm:passive-overview} follows a constructive approach. In \cref{sect:outline-analysis}, we present three explicit $F$-valid \schemes: one that maximizes producer surplus (\cref{sect:passive-ps-max}), one that maximizes consumer surplus (\cref{sect:passive-cs-max}), and one that minimizes social welfare (\cref{sect:passive-sw-min}). These correspond respectively to the top vertex (seller-optimal point), the bottom-right vertex (buyer-optimal point), and the bottom-left vertex (welfare-minimal point) of the red triangle in \cref{fig:summary}. The full proofs of the constructions are given in \cref{sect:proof-passive}. By taking convex combinations of these three extreme-point \schemes, we show that any point within the triangle can be implemented via some $F$-valid \scheme. This establishes the claimed region and completes the proof of \cref{thrm:passive-overview}.

\subsection{Outline of Analysis} \label{sect:outline-analysis}

Our constructions are based on the following definition of equal-revenue markets.

\begin{definition}[Equal-Revenue Market]
    A market $\bx \in \R_{\ge 0}^n$ is called an equal-revenue market if the seller obtains the same revenue at every price in its support, i.e., there exists a constant $C$ such that for all $p \in \support(\bx)$, $R_{\bx}(p) = C$.
\end{definition}

Given any price set $D \subseteq V$ (not necessarily contiguous), the unique equal-revenue market with unit total mass supported on $D$, denoted by $\bx^D = (x^D_1, x^D_2, \dots, x^D_n)$, is defined as follows:
\begin{equation} \label{eq:equal-revenue}
    x^D_i = 
    \begin{cases}
        0 & \text{if } v_i \not\in D, \\
        \min(D)/\max(D) & \text{if } v_i = \max(D), \\
        \min(D)\left(1 / v_i - 1 / \min\{v' \in D: v' > v_i\}\right) & \text{otherwise}.
    \end{cases}
\end{equation}
We further define $\ERM(D)$ as the set of all equal-revenue markets supported on $D$:
\[
    \ERM(D) = \{ \gamma \cdot \bx^D : \gamma \ge 0 \}.
\]

With this definition, \citet{bergemann2015limits} construct the consumer-surplus-maximizing \scheme in the absence of regulation. Specifically, for a market $\bx^*$, they iteratively identify the equal-revenue market $\bx^{\equal}$ supported on $\support(\bx^*)$ with the largest total mass that is dominated by $\bx^*$. The corresponding price for the market segment $\bx^{\equal}$ is set to the minimum value in its support:
\begin{equation} \label{eq:construction-bbm}
    \tag{No-Regulation}
    \left\{
    \begin{aligned}
        \bx^{\equal} &= \argmax \left\{ \mass(\bx) : \bx \in \ERM(\support(\bx^*)),\, \bx \le \bx^* \right\}, \\
        p &= \min(\support(\bx^{\equal})).
    \end{aligned}
    \right.
\end{equation}
This procedure is then applied recursively to the residual market $\bx^* - \bx^{\equal}$, continuing until the entire mass is exhausted.

\Cref{tab:construction-bbm} presents the result of applying the above method to the setting in \cref{ex:intro_four}. In the first round, the constructed equal-revenue market is supported on $\support(\bx^*) = \{1, 2, 3, 6\}$, as shown in the fourth column. The equal-revenue market with the maximum mass over this support is listed in the last column, along with the associated price determined by \cref{eq:construction-bbm}. This process continues for four rounds, ultimately reducing the market to zero. The four constructed market segments shown in the last column collectively constitute the final \scheme.

\begin{table}[htbp]
    \centering
    \caption{$V = \{1, 2, 3, 6\}$, $F = \{2, 3\}$, and $\bx^* = (0.36, 0.20, 0.18, 0.26)$. No regulation.}
    \vspace{-5px}
    \label{tab:construction-bbm}
    \scalebox{0.85}{
    \begin{tabular}{ccccc}
        \toprule
        Step & $\bx^*$ & $\optPrice(\bx^*)$ & $\support(\bx^{\equal}) = \support(\bx^*)$ & $(\bx^{\equal}, p)$ \\
        \midrule
        1 & $(0.36, 0.20, 0.18, 0.26)$ & \{6\} & \{1, 2, 3, 6\} & $((0.36, 0.12, 0.12, 0.12), 1)$ \\
        2 & $(0.00, 0.08, 0.06, 0.14)$ & \{6\} & \{2, 3, 6\} & $((0.00, 0.06, 0.06, 0.06), 2)$ \\
        3 & $(0.00, 0.02, 0.00, 0.08)$ & \{6\} & \{2, 6\} & $((0.00, 0.02, 0.00, 0.01), 2)$ \\
        4 & $(0.00, 0.00, 0.00, 0.07)$ & \{6\} & \{6\} & $((0.00, 0.00, 0.00, 0.07), 6)$ \\
        \midrule
        5 & $(0.00, 0.00, 0.00, 0.00)$ & \multicolumn{3}{c}{Done. Producer surplus: $1.56$. Consumer surplus: $1.30$.} \\
        \bottomrule
    \end{tabular}
    }
\end{table}

Note that although this method maximizes consumer surplus, the last segment in the example is supported solely on $\{6\}$, which is disjoint from $F$. As a result, the resulting \scheme is not \Fvalid, illustrating a limitation of the approach in \citet{bergemann2015limits} under regulatory constraints. A more detailed explanation of why the method of \citet{bergemann2015limits} fails under regulation is provided in \cref{app:bbm}.

Our constructions of different \schemes under regulation are inspired by the equal-revenue market framework. In each round, we iteratively extract a pair $(\bx^{\equal}, p)$ from the remaining market $\bx^*$, following the same high-level strategy. However, the specific design of $(\bx^{\equal}, p)$ differs across objectives, depending on which aspect—producer surplus, consumer surplus, or social welfare—is being optimized. A unified pseudocode summarizing these constructions is provided in \cref{alg:construction}.

\subsubsection{Producer-Surplus-Maximizing \SCHEME} \label{sect:passive-ps-max}
The high-level idea of constructing a producer-surplus-maximizing \scheme under regulation is similar to the no-regulation setting: segment consumers based on their values and charge them exactly their value. However, since prices must be constrained to lie in $F$, this idea can only be directly applied to consumers whose values fall within $F$. For consumers with values outside $F$, the segmentation rule under the \Fvalid constraint becomes non-trivial.

To address this, we construct the \scheme greedily by building market segments priced at values in $F$, enumerated from highest to lowest. At each step, the goal is to maximize the producer surplus extractable from the current market $\bx^*$ by including as many high-value consumers as possible subject to the \Fvalid constraint. To implement this, we modify the support of $\bx^{\equal}$ in \cref{eq:construction-bbm}, replacing $\support(\bx^*)$ with:
\begin{equation} \label{eq:S}
    B(\bx^*, F) = (\support(\bx^*) \backslash F) \cup \{\max(\support(\bx^*) \cap F)\}.
\end{equation}
That is, $B(\bx^*, F)$ includes all values outside $F$ and retains only the highest value within $F$. Based on this, $\bx^{\equal}$ and $p$ are constructed as follows:
\begin{equation} \label{eq:construction-PS-max}
    \tag{Regulation-PS-Max}
    \left\{
    \begin{aligned}
        \bx^{\equal} & = \argmax\{\mass(\bx): \bx \in \ERM(B(\bx^*, F)),\, \bx \le \bx^* \}, \\
        p & = \min(\support(\bx^{\equal}) \cap F).
    \end{aligned}
    \right.
\end{equation}

A detailed example based on \cref{ex:intro_four} is given in \cref{tab:construction-ps-max}. The key difference from \cref{tab:construction-bbm} lies in the fourth column $\support(\bx^{\equal})$, set to $B(\bx^*, F)$ rather than $\support(\bx^*)$. The resulting producer and consumer surpluses are $1.64$ and $0.86$, respectively, summing to $\PassiveSWMax(\bx^*, F) = 2.50$.

\begin{table}[htbp]
    \centering
    \caption{$V = \{1, 2, 3, 6\}$, $F = \{2, 3\}$, and $\bx^* = (0.36, 0.20, 0.18, 0.26)$. PS-maximizing under regulation.}
    \vspace{-5px}
    \label{tab:construction-ps-max}
    \scalebox{0.85}{
    \begin{tabular}{ccccc}
        \toprule
        Step & $\bx^*$ & $\optPrice(\bx^*)$ & $\support(\bx^{\equal}) = B(\bx^*, F)$ & $(\bx^{\equal}, p)$ \\
        \midrule
        1 & $(0.36, 0.20, 0.18, 0.26)$ & \{6\} & \{1, 3, 6\} & $((0.36, 0.00, 0.09, 0.09), 3)$ \\
        2 & $(0.00, 0.20, 0.09, 0.17)$ & \{6\} & \{3, 6\} & $((0.00, 0.00, 0.09, 0.09), 3)$ \\
        3 & $(0.00, 0.20, 0.00, 0.08)$ & \{2\} & \{2, 6\} & $((0.00, 0.16, 0.00, 0.08), 2)$ \\
        4 & $(0.00, 0.04, 0.00, 0.00)$ & \{2\} & \{2\} & $((0.00, 0.04, 0.00, 0.00), 2)$ \\
        \midrule
        5 & $(0.00, 0.00, 0.00, 0.00)$ & \multicolumn{3}{c}{Done. Producer surplus: $1.64$. Consumer surplus: $0.86$.} \\
        \bottomrule
    \end{tabular}}
\end{table}

We prove in \cref{sec:proof-passive-PS-PS} that this construction achieves our greedy objective: at each step, it extracts the maximum producer surplus from the residual market, using a carefully structured inductive argument. We further show in \cref{sec:proof-passive-PS-feasibility,sec:proof-passive-PS-CS} that this greedy method ensures feasibility and minimizes consumer surplus. In \cref{sect:justification}, we demonstrate that the top-down enumeration of $F$ is, in fact, the \emph{only} feasible strategy.

\subsubsection{Consumer-Surplus-Maximizing \SCHEME} \label{sect:passive-cs-max}
To maximize consumer surplus, the producer surplus must remain fixed at the uniform pricing revenue. A key observation from \citet{bergemann2015limits} is that their method always constructs equal-revenue markets whose support includes the optimal price of the original market $\bx^*$. This ensures that the producer surplus does not increase, allowing the residual value to be fully passed to consumers. In contrast, the PS-maximizing scheme in \cref{sect:passive-ps-max} may generate equal-revenue markets (e.g., with support $\{2\}$ in the final step of \cref{tab:construction-ps-max}) that exclude the original optimal price (e.g., $6$), thereby reducing consumer surplus.

To resolve this, our construction maintains the original optimal price as optimal in all residual markets. Specifically, $\bx^{\equal}$ is constructed based on two cases: (1) If $\optPrice(\bx^*) \cap F = \emptyset$, we use $B(\bx^*, F)$ from \cref{eq:S}, with the additional constraint that the optimal price of the residual market $\bx^* - \bx$ must still include that of $\bx^*$. (2) Otherwise, we apply \citet{bergemann2015limits}’s method directly. Formally, the construction is:
\begin{equation} \label{eq:construction-CS-max}
    \tag{Regulation-CS-Max}
    \left\{
    \begin{aligned}
        \bx^{\equal} & = \left\{\begin{aligned}
            & \argmax\{\mass(\bx): \bx \in \ERM(B(\bx^*, F)),\, \bx \le \bx^*,\, \optPrice(\bx^* - \bx) \supseteq \optPrice(\bx^*) \} \\
            & \qquad\qquad\qquad\qquad\qquad\qquad\qquad\qquad\qquad\qquad \text{if } \optPrice(\bx^*) \cap F = \emptyset, \\
            & \argmax\{\mass(\bx): \bx \in \ERM(\support(\bx^*)),\, \bx \le \bx^* \} \qquad\quad\text{otherwise.}
        \end{aligned}\right. \\
        p & = \min(\support(\bx^{\equal}) \cap F).
    \end{aligned}
    \right.
\end{equation}

\cref{tab:construction-cs-max} presents the result of applying our method to \cref{ex:intro_four}. The construction begins to diverge from \cref{tab:construction-ps-max} in the second row, where a smaller $\bx^{\equal}$ is selected to preserve $6$ as the optimal price of $\bx^* - \bx^{\equal}$. From that point onward, $6$ remains optimal in all residual markets and thus appears in the support of every constructed equal-revenue market. The resulting consumer surplus is $0.94$, larger than that in \cref{tab:construction-ps-max}, while the total surplus remains $\PassiveSWMax(\bx^*, F) = 2.50$.

\begin{table}[htbp]
    \centering
    \caption{$V = \{1, 2, 3, 6\}$, $F = \{2, 3\}$, and $\bx^* = (0.36, 0.20, 0.18, 0.26)$. CS-maximizing under regulation.}
    \vspace{-5px}
    \label{tab:construction-cs-max}
    \scalebox{0.85}{
    \begin{tabular}{ccccc}
        \toprule
        Step & $\bx^*$ & $\optPrice(\bx^*)$ & $\support(\bx^{\equal}) = B(\bx^*, F) \text{ or } \support(\bx^*)$ & $(\bx^{\equal}, p)$ \\
        \midrule
        1 & $(0.36, 0.20, 0.18, 0.26)$ & \{6\} & \{1, 3, 6\} ($B(\bx^*, F)$) & $((0.36, 0.00, 0.09, 0.09), 3)$ \\
        2 & $(0.00, 0.20, 0.09, 0.17)$ & \{6\} & \{3, 6\} ($B(\bx^*, F)$) & $((0.00, 0.00, 0.05, 0.05), 3)$ \\
        3 & $(0.00, 0.20, 0.04, 0.12)$ & \{2, 6\} & \{2, 3, 6\} ($\support(\bx^*)$) & $((0.00, 0.04, 0.04, 0.04), 2)$ \\
        4 & $(0.00, 0.16, 0.00, 0.08)$ & \{2, 6\} & \{2, 6\} ($\support(\bx^*)$) & $((0.00, 0.16, 0.00, 0.08), 2)$ \\
        \midrule
        5 & $(0.00, 0.00, 0.00, 0.00)$ & \multicolumn{3}{c}{Done. Producer surplus: $1.56$. Consumer surplus: $0.94$.} \\
        \bottomrule
    \end{tabular}}
\end{table}

From \cref{tab:construction-ps-max}, we observe an important switching point in the choice of $\support(\bx^{\equal})$: it changes from $B(\bx^*, F)$ to $\support(\bx^*)$ between Step~2 and Step~3, marking the transition from a seller-optimal to a buyer-optimal strategy. This switch is a key component of our construction. We prove in \cref{sect:proof-passive-CS} that this procedure ensures feasibility while simultaneously maximizing consumer surplus and minimizing producer surplus.

\subsubsection{Social-Welfare-Minimizing \SCHEME} \label{sect:passive-sw-min}
To minimize social welfare while preserving feasibility, we refine the regulated price set $F$ by raising its lower bound and limiting the use of certain values. This yields a minimal effective set $\tF$, interpreted as the smallest price set sufficient to maintain feasibility, as formalized in \cref{defn:sub-feasible} and illustrated by the shaded region on the left side of \cref{fig:proof-SWMin}. For instance, in \cref{fig:proof-SWMin}, we tighten the lower bound of $F$ from $v_2$ to $v_3$ and further restrict the total mass that can be priced at $v_3$.

\begin{figure}[htbp]
    \centering
        \includegraphics[width=0.7\linewidth]{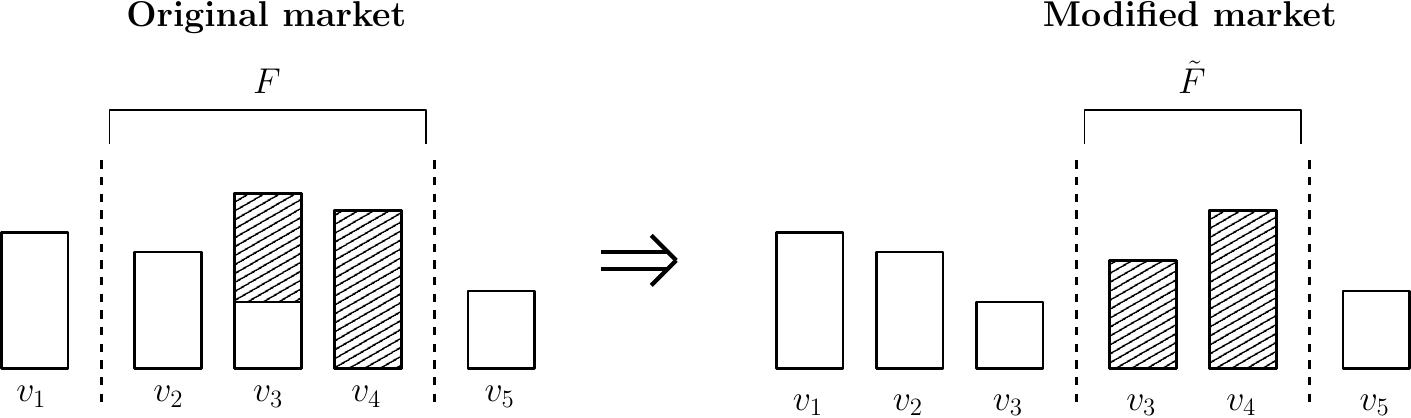}
    \vspace{-5px}
    \caption{Illustration of the idea behind constructing SW-minimizing \scheme.}
    \label{fig:proof-SWMin}
\end{figure}

Once $\tF$ is identified, we give an intuitive construction to illustrate the core idea: we remove $v_2$ from $F$ and split $x^*_3$ into two parts. The segmentation then proceeds using the greedy method from \cref{sect:passive-ps-max}, subject to the updated constraints imposed by $\tF$. This illustrative process is shown in the right panel of \cref{fig:proof-SWMin}. The formal construction and full analysis are provided in \cref{sect:proof-passive-SW}, where we prove that the resulting scheme minimizes both consumer and producer surplus.

\Cref{tab:construction-SW-min} shows the result of applying our method to \cref{ex:intro_four}. Only $0.16$ mass at $v_2$ is needed to satisfy feasibility, allowing $0.04$ to be extracted in the first row. The remaining rows follow the standard construction in \cref{eq:construction-PS-max}. The resulting producer surplus is $1.56$, matching that of the CS-maximizing \scheme in \cref{tab:construction-cs-max}, while the consumer surplus is $0.86$, matching that of the PS-maximizing \scheme in \cref{tab:construction-ps-max}.

\begin{table}[htbp]
    \centering
    \caption{$V = \{1, 2, 3, 6\}$, $F = \{2, 3\}$, and $\bx^* = (0.36, 0.20, 0.18, 0.26)$. SW-minimizing under regulation.}
    \vspace{-5px}
    \label{tab:construction-SW-min}
    \scalebox{0.85}{
    \begin{tabular}{ccccc}
        \toprule
        Step & $\bx^*$ & $\optPrice(\bx^*)$ & $\support(\bx^{\equal})$ & $(\bx^{\equal}, p)$ \\
        \midrule
        1 & $(0.36, 0.20, 0.18, 0.26)$ & \{6\} & \{1, 2, 3, 6\} & $((0.12, 0.04, 0.04, 0.04), 3)$ \\
        2 & $(0.24, 0.16, 0.14, 0.22)$ & \{6\} & \{1, 3, 6\} & $((0.24, 0.00, 0.06, 0.06), 3)$ \\
        3 & $(0.00, 0.16, 0.08, 0.16)$ & \{6\} & \{3, 6\} & $((0.00, 0.00, 0.08, 0.08), 3)$ \\
        4 & $(0.00, 0.16, 0.00, 0.08)$ & \{2, 6\} & \{2, 6\} & $((0.00, 0.16, 0.00, 0.08), 2)$ \\
        \midrule
        5 & $(0.00, 0.00, 0.00, 0.00)$ & \multicolumn{3}{c}{Done. Producer surplus: $1.56$. Consumer surplus: $0.86$.} \\
        \bottomrule
    \end{tabular}}
\end{table}



\section{Conclusions}
In this paper, we consider third-degree price discrimination with regulation by posting interval restriction on the seller's prices. For both passive and active intermediary scenarios, we characterize the feasible space of buyer and seller surplus by algorithmically constructing market segmentations. More importantly, we show that in both scenarios, we can achieve social welfare maximization while minimizing the seller's benefit, hence giving all the additional surplus to the buyers. 


\bibliographystyle{plainnat}
\bibliography{reference}

\clearpage

\appendix

\section{Proofs of the Passive Intermediary Model} \label{sect:proof-passive}
\paragraph{Additional Preliminaries.} \citet{CummingsD0W20} show that any market scheme can be converted to one where no two markets share the same instructed price, while preserving producer and consumer surplus. We capture this in the following definition.

\begin{definition} [Standard-form \SCHEME, \cite{CummingsD0W20}] \label{defn:standard-form}
    An \Fvalid \scheme $\Z$ is \emph{standard-form} if and only if $|Q| = |F|$, and $p_q = v_{\ell + q - 1}$ for any $q \in [|F|]$. Given an \Fvalid \scheme $\Z$, define the corresponding standard-form \scheme of $\Z = \{(\bx_q, p_q)\}_{q\in [|Q|]}$ as
    $$
        \standardization(\Z) = \{(\bx'_{q'}, p'_{q'})\}_{q' \in [|F|]} \text{ where } p'_{q'}  = v_{\ell + q' - 1}, \bx'_{q'} = \sum_{(\bx_q, p_q) \in \Z} \bbI[p_q = p'_{q'}]\cdot \bx_q.
    $$
\end{definition}

It therefore suffices to focus on \standard \schemes. 

\paragraph{A Unified Pseudo-code}
To support rigorous analysis, we present a unified pseudo-code in \cref{alg:construction} that captures the core ideas outlined in \cref{sect:outline-analysis} for constructing various \schemes. The only difference across these \schemes lies in \cref{line:construction-equation}, which invokes \cref{eq:construction-PS-max} for the producer-surplus-maximizing \scheme, \cref{eq:construction-CS-max} for the consumer-surplus-minimizing \scheme, and \cref{eq:construction-PS-max-modified-1,eq:construction-PS-max-modified-2} for the social-welfare-minimizing \scheme.

\begin{algorithm}[htbp]
    \caption{Constructive Method for Achieving \SCHEMES}
    \label{alg:construction}
    \begin{algorithmic}[1]
        \State \textbf{Input:} $\bx^*$, $F$, $V$
        \State $\Z \gets \{\}$;
        \While{$\support(\bx^*) \cap F \ne \emptyset$} \label{line:while}
            \State $(\bx^{\equal}, p)$ is determined by one of the specified construction rules \label{line:construction-equation}
            \Statex \Comment{\cref{eq:construction-PS-max,eq:construction-CS-max,eq:construction-PS-max-modified-1,eq:construction-PS-max-modified-2}}
            \State $\bx^* \leftarrow \bx^* - \bx^{\equal}$; \label{line:construction-xstar-update}
            \State $\Z \leftarrow \Z \cup \{(\bx^{\equal}, p)\}$
        \EndWhile
        \State $\bx^{\remain} \gets \bx^*$; \Comment{We will later show that $\bx^{\remain} = \boldsymbol{0}$} \label{line:bx-remain}
        \State $\Z \gets \standardization(\Z)$;
        \State \textbf{Output:} \scheme $\Z$
    \end{algorithmic}
\end{algorithm}

We now present rigorous proofs for the \schemes maximizing producer surplus, maximizing consumer surplus, and minimizing social welfare, in \cref{sect:proof-passive-PS,sect:proof-passive-CS,sect:proof-passive-SW}, respectively.

\subsection{Producer-Surplus-Maximizing \SCHEME} \label{sect:proof-passive-PS}
Let $\ZPP$ be the output of \cref{alg:construction} to achieve the PS-maximizing \scheme. The following theorem holds.
\begin{theorem} \label{thrm:passive-PS}
    Given a market $\bx^*$ and an $\bx^*$-feasible contiguous price set $F = \{v_\ell, \ldots, v_r\}$, let $\ZPP$ denote the output and let $\bx^{\remain}$ be defined as in \cref{line:bx-remain} of \cref{alg:construction}, when \cref{eq:construction-PS-max} is used in \cref{line:construction-equation}. Then $\bx^{\remain} = 0$ so that $\ZPP$ is an \Fvalid \scheme on $\bx^*$. Further,
    \begin{equation*} \label{eq:passive-PS}
        \CS(\ZPP) = \PassiveCSMin(\bx^*, F), \quad \PS(\ZPP) = \sum_{i=\ell}^n v_ix^*_i - \PassiveCSMin(\bx^*, F).
    \end{equation*}
\end{theorem}

\paragraph{Proof Sketch of \cref{thrm:passive-PS}.}
At a high level, we first show in \cref{sec:proof-passive-PS-PS} that even without requiring $\bx^{\remain} = \boldsymbol{0}$, the producer surplus of $\ZPP$ can exceed that of any \Fvalid \scheme of $\bx^*$. Building on this, \cref{sec:proof-passive-PS-feasibility} establishes that $\bx^{\remain}$ is indeed zero, and thus $\ZPP$ constitutes an \Fvalid \scheme of $\bx^*$. Finally, in \cref{sec:proof-passive-PS-CS}, we show that $\ZPP$ achieves the minimal consumer surplus among all \Fvalid \schemes—i.e., $\PassiveCSMin(\bx^*, F)$—and that its total welfare is exactly $\sum_{i = \ell}^n v_i x_i^*$.
These analyses together complete the proof of \cref{thrm:passive-PS}.

\subsubsection{Proof of \cref{thrm:passive-PS}: Producer Surplus Maximization} \label{sec:proof-passive-PS-PS}
In this part, we allow $\bx^{\remain}$ to be nonzero, and we slightly overload notation by letting $\ZPP = \{(\bx_q, p_q)\}_{q \in [|F|]}$ denote the \Fvalid \scheme on $\bx^* - \bx^{\remain}$, expressed in standard form as defined in \cref{defn:standard-form}. Note that the total mass of $\ZPP$ may be strictly less than that of $\bx^*$, and recall that each market $\bx_q$ has instructed price $p_q = v_{q + \ell - 1}$.

\begin{proposition} \label{thrm:passive-ps-ps}
For any \Fvalid \scheme $\Z_0$ of $\bx^*$, we have
\[
\PS(\ZPP) \ge \PS(\Z_0).
\]
\end{proposition}

\begin{proof}
Without loss of generality, we may assume that $\Z_0$ is in standard form, according to \cref{defn:standard-form}. Let $\Z_0 = \{(\bx'_q, p_q)\}_{q \in [|F|]}$ with $p_q = v_{\ell + q - 1}$. For any \standard \scheme $\Z$ and any $k \in [|F|]$, we introduce the notation $\PS_k(\Z)$ to denote the total producer surplus obtained from market segments in $\Z$ with instructed prices no less than $v_{r - k + 1}$, i.e.,
\[
\PS_k(\Z) = \sum_{(\bx, p) \in \Z} \ps(\bx, p) \cdot \bbI[p \ge v_{r - k + 1}].
\]
By definition, we have $\PS_{|F|}(\Z) = \PS(\Z)$. Therefore, to prove \cref{thrm:passive-ps-ps}, it suffices to show that $\PS_k(\ZPP) \ge \PS_k(\Z)$ for all $k \in [|F|]$. We establish this via induction on $k$.

The statement is always true when $k = 0$. By inductive hypothesis, assume that the statement is true when $k \le s$ ($s \ge 0$). Consider the case when $k = s + 1$.

Recall that $\ZPP = \{(\bx_q, p_q)\}_{q \in [|F|]}$ and $\Z_0 = \{(\bx'_q, p_q)\}_{q \in [|F|]}$, where $p_q = v_{\ell + q - 1}$. Define
\[
\hat{\bx} = \sum_{q = |F| - k + 1}^{|F|} \bx_q \quad \text{and} \quad \ui = \min\{i > r : \hat{x}_{i} < x_i^*\}.
\]
If there is no such index $i$ satisfying $\hat{x}_{i} < x_i^*$, we define $\ui = n + 1$. In addition, for each $q \in \{|F| - k + 1, \ldots, |F|\}$, define
\[
L_q = \sum_{i = \ell + q - 1}^{\ui - 1} x_{q,i}, \quad S_q = G_{\bx_q}(v_{\ui}), \quad L_q' = \sum_{i = \ell + q - 1}^{\ui - 1} x_{q,i}', \quad S_q' = G_{\bx_q'}(v_{\ui}), \quad D_q = L_q - L_q'.
\]
We set $S_q = 0$ and $S_q' = 0$ if $\ui = n + 1$. $v_{n+1}$ is set as $+\infty$. The intuition behind this construction is as follows. For each $\bx_q$, we divide its mass over $\{v_{\ell + q - 1}, \ldots, v_n\}$ into two parts: $L_q$ captures the mass below $v_{\ui}$ (i.e., over $\{v_{\ell + q - 1}, \ldots, v_{\ui - 1}\}$), while $S_q$ captures the remaining mass at or above $v_{\ui}$. The same decomposition is applied to $\bx_q'$, yielding $L_q'$ and $S_q'$. The definitions of $L_q$, $S_q$, $L_q'$, $S_q'$, and $D_q$ form the basis of the argument. We begin by stating several basic observations, whose proofs are provided later.
\begin{lemma} \label{lemma:c1}
The following three statements hold:
\begin{enumerate}
    \item For all $q \in \{|F|-k+1,\ldots,|F|\}$, it holds that $S_q = L_q  \cdot \frac{p_{q}}{v_{\ui} - p_q}$.
    \item For all $q \in \{|F|-k+1,\ldots,|F|\}$, it holds that $S_q' \le L_q' \cdot \frac{p_{q}}{v_{\ui} - p_q}$.
    \item $\sum_{q = |F|-k+1}^{|F|} D_q \ge 0$.
\end{enumerate}
\end{lemma}

Define
\[
t^* = \max\left\{0 < t \le k : \sum_{q = |F| - t + 1}^{|F|} D_q < 0\right\}.
\]
If no such $t$ exists, set $t^* = 0$. By the third part of \cref{lemma:c1}, we must have $t^* < k$. It follows that for any $t' \in [t^* + 1, k]$, the partial sum satisfies
\[
\sum_{q = |F| - t' + 1}^{|F|} D_q \ge 0.
\]
Therefore, we obtain:
\begin{equation} \label{eq:D}
    \forall t' \in [t^* + 1, k], \qquad 
    \sum_{q = |F| - t' + 1}^{|F| - t^*} D_q 
    = \underbrace{\left(\sum_{q = |F| - t' + 1}^{|F|} D_q\right)}_{\ge 0}
    - \underbrace{\left(\sum_{q = |F| - t^* + 1}^{|F|} D_q\right)}_{\le 0} 
    \ge 0.
\end{equation}
In addition, since $L_q + S_q = G_{\bx_q}(p_q)$  and $L'_q + S'_q = G_{\bx'_q}(p_q)$, we have for any $t \le k$: 
$$
\PS_t(\ZPP) = \sum_{q=|F| - t + 1}^{|F|} p_q\left(L_q  + S_q \right) \quad \mbox{and} \quad \PS_t(\Z_0) = \sum_{q=|F| - t + 1}^{|F|} p_q\left(L'_q  + S'_q \right).
$$
By the induction hypothesis, we have $\PS_{t^*}(\ZPP) \ge \PS_{t^*}(\Z_0)$ since $t^* < k$. We have the following chain of inequalities: 
\begin{align*}
    & \, \PS_{k}(\ZPP) -  \PS_{k}(\Z_0) \\
    = & \, \PS_{t^*}(\ZPP) - \PS_{t^*}(\Z_0) + \sum_{q = |F| - k + 1} ^ {|F| - t^*} p_q(L_{q}  + S_{q}  - L_{q}' - S_{q}' ) \\
    \ge & \, 0 + \sum_{q = |F|-k + 1}^{|F|-t^*} p_q\left(L_q \cdot \frac{v_{\ui}}{v_{\ui} - p_q} - L_q' \cdot \frac{v_{\ui}}{v_{\ui} - p_q} \right) \\
    = & \, \sum_{q = |F|-k+ 1} ^ {|F|-t^*} (L_q  - L_q')\cdot \frac{p_qv_{\ui}}{v_{\ui} - p_q}  = \sum_{q = |F|-k+ 1} ^ {|F|-t^*} D_q\cdot \frac{p_qv_{\ui}}{v_{\ui} - p_q} \\
    = & \, \frac{p_{|F|-k+1}v_{\ui}}{v_{\ui} - p_{|F|-k+1}} \cdot \sum_{q = |F|-k+ 1} ^ {|F|-t^*} D_q + \sum_{q = |F| - k + 2}^{|F|-t^*} \left(\sum_{u = q}^{|F|-t^*} D_u \right) \cdot \left( \frac{p_qv_{\ui}}{v_{\ui} - p_q} - \frac{p_{q-1}v_{\ui}}{v_{\ui} - p_{q-1}}\right) \ge  0.
\end{align*}
Here, the first inequality follows from the induction hypothesis and \cref{lemma:c1}; The last inequality follows from \cref{eq:D} and the fact that $p_qv_{\ui}/(v_{\ui} - p_q)$ is increasing with respect to $q$. Therefore, the statement holds for $k = s+1$, completing the proof.
\end{proof}

\paragraph{Proof of \cref{lemma:c1}.} We now prove \cref{lemma:c1}.
\begin{proof}
    For the first point, consider the case when $\ui = n + 1$. By definition, we have $S_q = 0$ and $v_{\ui} = +\infty$, so the claim holds trivially. Now consider the case when $\ui < n + 1$. By construction of $\ui$, the mass $x_i^*$ is not fully exhausted after removing the market segments $\bx_{|F| - k + 1}, \bx_{|F| - k + 2}, \ldots, \bx_{|F|}$. According to \cref{alg:construction} and \cref{eq:construction-PS-max}, each equal-revenue market $\bx^{\equal}$ generated in \cref{line:construction-equation} of \cref{alg:construction} has support that includes $v_{\ui}$. As a result, it follows that $R_{\bx_q}(v_{\ui}) = R_{\bx_q}(p_q)$ for all $q \in \{|F| - k + 1, \ldots, |F|\}$. Hence,
    \begin{equation} \label{eq:proof-lemma-c1}
        v_{\ui} S_q = R_{\bx_q}(v_{\ui}) = R_{\bx_q}(p_q) = p_q(L_q + S_q),
    \end{equation}
    which shows $S_q = p_q L_q / (v_{\ui} - p_q)$.

    For the second point, since the instructed price for each $\bx_q'$ with $q \in \{|F| - k + 1, \ldots, |F|\}$ is $p_q$, it follows that $R_{\bx_q'}(v_{\ui}) = R_{\bx_q'}(p_q)$. Repeating the derivation in \cref{eq:proof-lemma-c1}, we obtain
    \[
    S_q' \le L_q' \cdot \frac{p_q}{v_{\ui} - p_q}.
    \]

    For the third point, by the definition of $\ui$ and the design of \cref{alg:construction} together with \cref{eq:construction-PS-max}, the construction ensures that $\ZPP$ exhausts all probability mass in $\bx^*$ over the interval $[p_{|F| - k + 1}, v_{\ui})$ via the markets $\{\bx_{|F| - k + 1}, \ldots, \bx_{|F|}\}$. This implies that
    \[
    \sum_{q = |F| - k + 1}^{|F|} L_q = \sum_{i = r - k + 1}^{\ui - 1} x_i^*.
    \]
    By the definition of $L_q'$, the total lower mass allocated in $\Z_0$ cannot exceed the available mass on values below $v_{\ui}$, and thus
    \[
    \sum_{q = |F| - k + 1}^{|F|} L_q \ge \sum_{q = |F| - k + 1}^{|F|} L_q',
    \]
    which implies $\sum_{q = |F| - k + 1}^{|F|} D_q \ge 0$.
\end{proof}

\subsubsection{Proof of \cref{thrm:passive-PS}: Feasibility} \label{sec:proof-passive-PS-feasibility}
Using the optimality of $\PS$, we have the following lemma showing that $\bx^{\remain} = \boldsymbol{0}$, so that the output of \cref{alg:construction} is \Fvalid.

\begin{proposition} \label{lemma:passive-PS-x-equals-0}
    $\bx^{\remain} = \boldsymbol{0}$, and thus $\ZPP$ constitutes an \Fvalid \scheme on $\bx^*$.
\end{proposition}
\begin{proof}
    Suppose, for contradiction, that $\bx^{\remain} \ne \boldsymbol{0}$, and let $\ZPP$ be the \Fvalid scheme on the residual market $\bx^* - \bx^{\remain}$. By construction, \cref{alg:construction} exhausts all probability mass on any value $v \in F$. Let $v' \in \support(\bx^{\remain})$. Since $v' \notin F$, \cref{eq:construction-PS-max} and \cref{alg:construction} ensure that the support of every equal-revenue market $\bx^{\equal}$ constructed in \cref{line:construction-equation} must include $v'$. Therefore, for all $(\by, p) \in \ZPP$, it holds that $R_{\by}(v') = R_{\by}(p)$.

    Since $F$ is $\bx^*$-feasible, let $\Z_1$ be any \Fvalid scheme on $\bx^*$. Then we have:
    \begin{align*}
        R_{\bx^*}(v') &= R_{\bx^{\remain}}(v') + \sum_{(\by, p) \in \ZPP} R_{\by}(v') \\
        &= R_{\bx^{\remain}}(v') + \sum_{(\by, p) \in \ZPP} R_{\by}(p) \\
        &= R_{\bx^{\remain}}(v') + \PS(\ZPP) \\
        &\ge R_{\bx^{\remain}}(v') + \PS(\Z_1) \tag{by \cref{thrm:passive-ps-ps}} \\
        &> \PS(\Z_1) \tag{since $v' \in \support(\bx^{\remain})$ and $R_{\bx^{\remain}}(v') > 0$} \\
        &\ge R_{\bx^*}(v').
    \end{align*}
    This yields a contradiction: $R_{\bx^*}(v') > R_{\bx^*}(v')$. Hence, it must be that $\bx^{\remain} = \boldsymbol{0}$.
\end{proof}

\subsubsection{Proof of \cref{thrm:passive-PS}: Consumer Surplus Minimization} \label{sec:proof-passive-PS-CS}
We now show that the output $\ZPP$ of \cref{alg:construction} minimizes consumer surplus among all regulated \schemes when \cref{eq:construction-PS-max} is used in \cref{line:construction-equation}. The result follows from the following lemma.

\begin{lemma}\label{lemma:proof-passive-PS-total-welfare}
    For any \Fvalid \scheme $\Z'$ of the market $\bx^*$, there exists an \Fvalid \scheme $\Z$ such that
    \[
    \PS(\Z) \ge \PS(\Z'), \quad \CS(\Z) = \CS(\Z'), \quad \text{and} \quad \PS(\Z) + \CS(\Z) = \sum_{i = \ell}^n v_i x_i^*.
    \]
\end{lemma}
\begin{proof}
    According to \cref{defn:standard-form}, without loss of generality, we may assume that $\Z'$ is in standard form and is given by $\Z' = \{(\bx'_q, p_q)\}_{q \in [|F|]}$, where $p_q = v_{\ell + q - 1}$. We now construct a new scheme $\Z$ based on $\Z'$ as follows.

    For each $(\bx'_q, p_q) \in \Z'$, we create $q$ new market segments $\{(\by_{q'}^{(q)}, p_{q'}^{(q)})\}_{q' \in [q]}$ defined by:
    \[
    \begin{aligned}
        y_{q',i}^{(q)} &= 
        \begin{cases}
            x'_{q,i}, & \text{if } i = \ell + q' - 1, \\
            0, & \text{otherwise}
        \end{cases}
        && \forall q' \in [q - 1],\ i \in [n], \\
        y_{q,i}^{(q)} &= 
        \begin{cases}
            x'_{q,i}, & \text{if } i \notin \{\ell, \ell + 1, \ldots, \ell + q - 1\}, \\
            0, & \text{otherwise}
        \end{cases}
        && \forall i \in [n], \\
        p_{q'}^{(q)} &= v_{\ell + q' - 1}.
    \end{aligned}
    \]
    It is easy to verify that
    \[
    \sum_{q' \in [q]} \by_{q'}^{(q)} = \bx'_q \quad \text{and} \quad p_{q'}^{(q)} \in \optPrice(\by_{q'}^{(q)}), \forall q' \in [q].
    \]
    In addition, by construction, we have
    \[
    \begin{aligned}
        & \ps(\by^{(q)}_{q'}, p^{(q)}_{q'}) = x'_{q, \ell+q'-1} \cdot v_{\ell+q'-1}, \quad \cs(\by^{(q)}_{q'}, p^{(q)}_{q'}) = 0, \quad \forall q' \in [q - 1] \\
        & \ps(\by^{(q)}_{q}, p^{(q)}_{q}) = \ps(\bx'_q, p_q) = v_{\ell+q-1}\cdot \sum_{i=\ell+q-1}^n x'_{q,i} \\
        & \cs(\by^{(q)}_{q'}, p^{(q)}_{q'}) = \cs(\bx'_q, p_q) = \sum_{i=\ell+q-1}^n x'_{q,i}v_i - v_{\ell+q-1}\cdot \sum_{i=\ell+q-1}^n x'_{q,i}.
    \end{aligned}
    \]
    As a result, we obtain:
    \[
    \begin{aligned}
        & \sum_{q' \in [q]} \ps(\by^{(q)}_{q'}, p^{(q)}_{q'}) \ge \ps(\bx'_q, p_q), \quad \sum_{q' \in [q]} \cs(\by^{(q)}_{q'}, p^{(q)}_{q'}) = \cs(\bx'_q, p_q) \\
        & \sum_{q' \in [q]} \left(\ps(\by^{(q)}_{q'}, p^{(q)}_{q'}) + \cs(\by^{(q)}_{q'}, p^{(q)}_{q'})\right) = \sum_{i=\ell}^n \bx'_{q,i}v_i.
    \end{aligned}
    \]
    Let $\Z$ denote the collection of all market segments $(\by_{q'}^{(q)}, p_{q'}^{(q)})$ for all $q \in [|F|]$ and $q' \in [q]$. Then the aggregate producer and consumer surplus under $\Z$ satisfy:
    \[
    \begin{aligned}
        \PS(\Z) & = \sum_{q \in |F|} \sum_{q' \in [q]} \ps(\by^{(q)}_{q'}, p^{(q)}_{q'}) \ge \sum_{q \in [|F|]} \ps(\bx'_q, p_q) = \PS(Z'), \\
        \CS(\Z) & = \sum_{q \in |F|} \sum_{q' \in [q]} \cs(\by^{(q)}_{q'}, p^{(q)}_{q'} = \sum_{q \in [|F|]} \cs(\bx'_q, p_q) = \CS(Z'), \\
        \PS(\Z) + \CS(\Z) & = \sum_{q \in |F|} \sum_{q' \in [q]} \left(\ps(\by^{(q)}_{q'}, p^{(q)}_{q'}) + \cs(\by^{(q)}_{q'}, p^{(q)}_{q'})\right) = \sum_{q \in |F|} \sum_{i=\ell}^n \bx'_{q,i}v_i = \sum_{i=\ell}^n x^*_i v_i.
    \end{aligned}
    \]
    This completes the construction, and the claim follows.
\end{proof}

Finally, we establish the consumer surplus and total welfare properties of $\ZPP$.

\begin{proposition}
    It holds that
    \[
    \CS(\ZPP) = \PassiveCSMin(\bx^*, F), \quad \PS(\ZPP) = \sum_{i=\ell}^n v_ix^*_i - \PassiveCSMin(\bx^*, F).
    \]
\end{proposition}
\begin{proof}
    By \cref{lemma:passive-PS-x-equals-0}, we have that $\ZPP$ is an \Fvalid \scheme of $\bx^*$. Moreover, from the construction in \cref{alg:construction} and the condition in \cref{eq:construction-PS-max}, it follows that in each market segment, consumers with valuations at least $v_{\ell}$ always purchase the good. Therefore, the total surplus generated by $\ZPP$ is:
    \[
    \SW(\ZPP) = \PS(\ZPP) + \CS(\ZPP) = \sum_{i = \ell}^n v_i x_i^*.
    \]
    
    Now consider any \Fvalid \scheme $\Z_0$ of $\bx^*$. By \cref{lemma:proof-passive-PS-total-welfare}, there exists an \Fvalid \scheme $\Z_1$ such that $\CS(\Z_1) = \CS(\Z_0)$ and $\PS(\Z_1) + \CS(\Z_1) = \sum_{i = \ell}^n v_i x_i^*$. Rearranging yields:
    \[
    \CS(\Z_0) = \CS(\Z_1) = \sum_{i = \ell}^n v_i x_i^* - \PS(\Z_1) \ge \sum_{i = \ell}^n v_i x_i^* - \PS(\ZPP) = \CS(\ZPP),
    \]
    where the inequality follows from \cref{thrm:passive-ps-ps}. This shows that the consumer surplus under $\ZPP$ is no greater than that under any other \Fvalid \scheme. Hence, $\CS(\ZPP) = \PassiveCSMin(\bx^*, F)$. This completes the proof.
\end{proof}

\subsection{Consumer-Surplus-Maximizing \SCHEME} \label{sect:proof-passive-CS}
Let $\ZPC$ be the output of \cref{alg:construction} to achieve the CS-maximizing \scheme. The following theorem holds.

\begin{theorem} \label{thrm:passive-CS}
    Given a market $\bx^*$ and an $\bx^*$-feasible contiguous price set $F = \{v_\ell, \ldots, v_r\}$, let $\ZPC$ denote the output and let $\bx^{\remain}$ be defined as in \cref{line:bx-remain} of \cref{alg:construction}, when \cref{eq:construction-CS-max} is used in \cref{line:construction-equation}. Then $\bx^{\remain} = 0$ so that $\ZPC$ is an \Fvalid \scheme on $\bx^*$. Further,
    \[
        \CS(\ZPC) = \sum_{i=\ell}^n v_ix^*_i - \UniRev(\bx^*), \quad \PS(\ZPC) = \UniRev(\bx^*).
    \]
\end{theorem}

\paragraph{Proof Sketch of \cref{thrm:passive-CS}.} 
We proceed in three steps. First, in \cref{sect:proof-passive-CS-basic}, we establish that the construction of $\ZPC$ satisfies the necessary conditions for minimizing consumer surplus. Then, \cref{sect:proof-passive-CS-feasibility} proves that $\ZPC$ constitutes a feasible \Fvalid \scheme. Finally, in \cref{sect:proof-passive-CS-final}, we show that $\ZPC$ achieves the targeted consumer and producer surplus guarantees.

\subsubsection{Proof of \cref{thrm:passive-CS}: Basic Properties} \label{sect:proof-passive-CS-basic}
We begin with the following lemmas, which characterize key properties of \cref{eq:construction-CS-max}, focusing on the relationship between $\optPrice(\bx^* - \bx^{\equal})$, $\optPrice(\bx^*)$, and the regulated price set $F$. The following lemmas cover all possible cases in \cref{eq:construction-CS-max}.

\begin{lemma} \label{lemma:proof-passive-CS-case-1}
    Suppose that $\optPrice(\bx^*) \cap F = \emptyset$, and define
    \[
    \bx^{\equal} = \argmax \left\{ \mass(\bx) : \bx \in \ERM(B(\bx^*, F)),\ \bx \le \bx^*,\ \optPrice(\bx^* - \bx) \supseteq \optPrice(\bx^*) \right\}
    \]
    (i.e., the first case in \cref{eq:construction-CS-max}). Suppose that $\bx^{\equal}$ is selected due to the binding constraint $\bx \le \bx^*$, in the sense that for any $\gamma > 0$, we have $(1+\gamma)\bx^{\equal} \not\le \bx^*$. Then it must hold that
    \[
    \optPrice(\bx^* - \bx^{\equal}) \supseteq \optPrice(\bx^*).
    \]
\end{lemma}
\begin{proof}
    By the constraint $\optPrice(\bx^* - \bx) \supseteq \optPrice(\bx^*)$ in \cref{eq:construction-CS-max}, and noting that the set 
    \[
    \left\{ \bx \in \ERM(B(\bx^*, F)) : \bx \le \bx^*, \optPrice(\bx^* - \bx) \supseteq \optPrice(\bx^*) \right\}
    \]
    is closed, the claim follows immediately.
\end{proof}

\begin{lemma} \label{lemma:proof-passive-CS-case-2}
    Suppose that $\optPrice(\bx^*) \cap F = \emptyset$, and define
    \[
    \bx^{\equal} = \argmax \left\{ \mass(\bx) : \bx \in \ERM(B(\bx^*, F)),\ \bx \le \bx^*,\ \optPrice(\bx^* - \bx) \supseteq \optPrice(\bx^*) \right\}
    \]
    (i.e., the first case in \cref{eq:construction-CS-max}). Suppose that $\bx^{\equal}$ is selected due to the binding constraint $\optPrice(\bx^* - \bx) \supseteq \optPrice(\bx^*)$, in the sense that for any $\gamma > 0$, we have $\optPrice(\bx^* - (1 + \gamma)\bx^{\equal}) \not\supseteq \optPrice(\bx^*)$. Then it must hold that
    \[
    \optPrice(\bx^* - \bx^{\equal}) \cap F \ne \emptyset \quad \text{and} \quad \optPrice(\bx^* - \bx^{\equal}) \supseteq \optPrice(\bx^*).
    \]
\end{lemma}
\begin{proof}

    Fix any $v_i \in \optPrice(\bx^*)$ and consider any $v_j \in \support(\bx^*)$. Observe the revenue gap:
    \[
    R_{\bx^* - (1 + \gamma)\bx^{\equal}}(v_j) - R_{\bx^* - (1 + \gamma)\bx^{\equal}}(v_i) 
    = R_{\bx^*}(v_j) - R_{\bx^*}(v_i) - (1 + \gamma)\left(R_{\bx^{\equal}}(v_j) - R_{\bx^{\equal}}(v_i)\right),
    \]
    which is a linear function of $\gamma$. In addition, note that $\optPrice(\bx^*) \subseteq B(\bx^*, F) \subseteq \support(\bx^*)$ by construction and the assumption $\optPrice(\bx^*)\cap F =\emptyset$. As a result, we observe the following:
    \begin{enumerate}
        \item If $v_j \in \optPrice(\bx^*)$, then $R_{\bx^*}(v_j) = R_{\bx^*}(v_i)$ and $R_{\bx^{\equal}}(v_j) = R_{\bx^{\equal}}(v_i)$. Therefore, $R_{\bx^* - (1 + \gamma)\bx^{\equal}}(v_j) = R_{\bx^* - (1 + \gamma)\bx^{\equal}}(v_i)$.
        \item If $v_j \in B(\bx^*, F) \setminus \optPrice(\bx^*)$, then $R_{\bx^*}(v_j) < R_{\bx^*}(v_i)$ while $R_{\bx^{\equal}}(v_j) = R_{\bx^{\equal}}(v_i)$. Consequently, $R_{\bx^* - (1 + \gamma)\bx^{\equal}}(v_j) < R_{\bx^* - (1 + \gamma)\bx^{\equal}}(v_i)$.
        \item If $v_j \in \support(\bx^*) \setminus B(\bx^*, F) \subseteq F$, then $R_{\bx^*}(v_j) < R_{\bx^*}(v_i)$ and $R_{\bx^{\equal}}(v_j) < R_{\bx^{\equal}}(v_i)$. Hence, there exists $\gamma_j$ such that:
        \[
        R_{\bx^* - (1 + \gamma)\bx^{\equal}}(v_j)
        \begin{cases}
            < R_{\bx^* - (1 + \gamma)\bx^{\equal}}(v_i), & \text{if } \gamma < \gamma_j, \\
            = R_{\bx^* - (1 + \gamma)\bx^{\equal}}(v_i), & \text{if } \gamma = \gamma_j, \\
            > R_{\bx^* - (1 + \gamma)\bx^{\equal}}(v_i), & \text{if } \gamma > \gamma_j.
        \end{cases}
        \]
        Note that $\optPrice(\bx^* - \bx^{\equal}) \supseteq \optPrice(\bx^*)$. Therefore, $R_{\bx^* - \bx^{\equal}}(v_j) \le R_{\bx^* - \bx^{\equal}}(v_i)$, which implies that $\gamma_j \ge 0$.
    \end{enumerate}
    For any $\gamma > 0$, since $\optPrice(\bx^* - (1 + \gamma)\bx^{\equal}) \not\supseteq \optPrice(\bx^*)$, there exists some $v_i \in \optPrice(\bx^*)$ such that $v_i \notin \optPrice(\bx^* - (1 + \gamma)\bx^{\equal})$. This in turn implies the existence of $v_j \in \support(\bx^*)$ with $R_{\bx^* - (1 + \gamma)\bx^{\equal}}(v_j) > R_{\bx^* - (1 + \gamma)\bx^{\equal}}(v_i)$. By the observations above, this situation can only occur in the third case, namely when there exists $v_j \in \support(\bx^*) \setminus B(\bx^*, F) \subseteq F$ with $\gamma_j < \gamma$. Taking the limit $\gamma \to 0^+$, we conclude that there must exist some $v_j \in F$ such that $\gamma_j = 0$.
    
    As a result, we have
    \[
    R_{\bx^* - (1 + \gamma)\bx^{\equal}}(v_j) > R_{\bx^* - (1 + \gamma)\bx^{\equal}}(v_i), \quad \forall v_i \in \optPrice(\bx^*),\ \forall \gamma > 0.
    \]
    Now consider $\gamma \to 0^+$. Since $\optPrice(\bx^* - \bx^{\equal}) \subseteq \optPrice(\bx^*)$ by construction, we must have:
    \[
    R_{\bx^* - \bx^{\equal}}(v_j) \le R_{\bx^* - \bx^{\equal}}(v_i), \quad \forall v_i \in \optPrice(\bx^*).
    \]
    But from the earlier inequality and continuity, we must also have equality:
    \[
    R_{\bx^* - \bx^{\equal}}(v_j) = R_{\bx^* - \bx^{\equal}}(v_i), \quad \forall v_i \in \optPrice(\bx^*),
    \]
    which implies $v_j \in \optPrice(\bx^* - \bx^{\equal})$ and thus $\optPrice(\bx^* - \bx^{\equal}) \cap F \ne \emptyset$.

    Finally, since the feasible set
    \[
    \left\{ \bx \in \ERM(B(\bx^*, F)) : \bx \le \bx^*,\ \optPrice(\bx^* - \bx) \supseteq \optPrice(\bx^*) \right\}
    \]
    is closed, the inclusion $\optPrice(\bx^* - \bx^{\equal}) \supseteq \optPrice(\bx^*)$ holds as well. This completes the proof.
\end{proof}

\begin{lemma} \label{lemma:proof-passive-CS-case-3}
    Suppose that $\optPrice(\bx^*) \cap F \ne \emptyset$, and define
    \[
    \bx^{\equal} = \argmax\left\{ \mass(\bx) : \bx \in \ERM(\support(\bx^*)),\ \bx \le \bx^* \right\}
    \]
    (i.e., the second case in \cref{eq:construction-CS-max}). Then:
    \[
    \optPrice(\bx^* - \bx^{\equal}) \cap F \ne \emptyset \quad \text{and} \quad \optPrice(\bx^* - \bx^{\equal}) \supseteq \optPrice(\bx^*).
    \]
\end{lemma}

\begin{proof}
    Let $v_i \in \optPrice(\bx^*) \cap F$. By the equal-revenue property (see \cref{eq:equal-revenue}), we have $R_{\bx^{\equal}}(v_i) = R_{\bx^{\equal}}(v_j)$ for all $v_j \in \support(\bx^*)$. Since $v_i \in \optPrice(\bx^*)$, it follows that for any $v_j \in \support(\bx^*)$,
    \[
    R_{\bx^* - \bx^{\equal}}(v_i) = R_{\bx^*}(v_i) - R_{\bx^{\equal}}(v_i) 
    \ge R_{\bx^*}(v_j) - R_{\bx^{\equal}}(v_j) = R_{\bx^* - \bx^{\equal}}(v_j).
    \]
    Therefore, $v_i \in \optPrice(\bx^* - \bx^{\equal})$, and we conclude that
    \[
    \optPrice(\bx^* - \bx^{\equal}) \cap F \supseteq \{v_i\} \ne \emptyset.
    \]

    Moreover, the inequality above holds for all $v_i \in \optPrice(\bx^*)$, implying that
    \[
    \optPrice(\bx^* - \bx^{\equal}) \supseteq \optPrice(\bx^*).
    \]
    This completes the proof.
\end{proof}

\subsubsection{Proof of \cref{thrm:passive-CS}: Feasibility} \label{sect:proof-passive-CS-feasibility}
\begin{proposition} \label{lemma:passive-CS-x-equals-0}
    $\bx^{\remain} = \boldsymbol{0}$, and hence $\ZPC$ constitutes an \Fvalid \scheme on $\bx^*$.
\end{proposition}

\begin{proof}
    Suppose, for contradiction, that $\bx^{\remain} \ne \boldsymbol{0}$. By construction in \cref{alg:construction}, it follows that $\support(\bx^{\remain}) \cap F = \emptyset$, and thus
    \[
    \optPrice(\bx^{\remain}) \cap F \subseteq \support(\bx^{\remain}) \cap F = \emptyset.
    \]

    According to \cref{lemma:proof-passive-CS-case-2,lemma:proof-passive-CS-case-3}, throughout the execution of \cref{alg:construction}, the equal-revenue market $\bx^{\equal}$ is never selected as either
    \[
    \argmax \left\{ \mass(\bx) : \bx \in \ERM(\support(\bx^*)),\ \bx \le \bx^* \right\}
    \]
    or
    \[
    \argmax \left\{ \mass(\bx) : \bx \in \ERM(B(\bx^*, F)),\ \bx \le \bx^*,\ \optPrice(\bx^* - \bx) \supseteq \optPrice(\bx^*) \right\},
    \]
    with the latter being selected due to the binding constraint $\optPrice(\bx^* - \bx) \supseteq \optPrice(\bx^*)$. Otherwise, it would follow from the respective lemmas that $\optPrice(\bx^{\remain}) \cap F \ne \emptyset$, contradicting our assumption.

    Therefore, in every iteration, the algorithm must have selected
    \[
    \bx^{\equal} = \argmax \left\{ \mass(\bx) : \bx \in \ERM(B(\bx^*, F)),\ \bx \le \bx^* \right\},
    \]
    which coincides exactly with the case specified in \cref{eq:construction-PS-max}. This implies that the behavior of the algorithm is identical to the version executed with \cref{eq:construction-PS-max} in \cref{line:construction-equation}. 

    By \cref{thrm:passive-PS}, we know that in this case the remainder must satisfy $\bx^{\remain} = \boldsymbol{0}$, contradicting our initial assumption. Hence, it must be that $\bx^{\remain} = \boldsymbol{0}$.
\end{proof}

\subsubsection{Proof of \cref{thrm:passive-CS}: Consumer Surplus Maximization and Producer Surplus Minimization} \label{sect:proof-passive-CS-final}
\begin{proposition} 
    It holds that
    \[
    \CS(\ZPC) = \sum_{i=\ell}^n v_i x^*_i - \UniRev(\bx^*), \quad \text{and} \quad \PS(\ZPC) = \UniRev(\bx^*).
    \]
\end{proposition}

\begin{proof}
    Let $v_i \in \optPrice(\bx^*)$ be any optimal price of the original market $\bx^*$. By \cref{lemma:proof-passive-CS-case-1,lemma:proof-passive-CS-case-2,lemma:proof-passive-CS-case-3}, throughout the execution of \cref{alg:construction}, the active set $A$ in the while-loop always satisfies $A \subseteq \optPrice(\bx^*)$. Therefore, in each iteration, the equal-revenue market $\bx^{\equal}$ selected by \cref{eq:construction-CS-max} must also satisfy $A \subseteq \bx^{\equal}$. This implies that every segment $(\by, p)$ in $\ZPC$ satisfies $R_{\by}(p) = R_{\by}(v_i)$ for some $v_i \in \optPrice(\bx^*)$. Hence,
    \[
    \PS(\ZPC) = \sum_{(\by, p) \in \ZPC} R_{\by}(p) = \sum_{(\by, p) \in \ZPC} R_{\by}(v_i) = \UniRev(\bx^*).
    \]

    Moreover, by construction and \cref{lemma:passive-CS-x-equals-0}, the algorithm fully allocates $\bx^*$, i.e., $\bx^{\remain} = \boldsymbol{0}$. Since each market segment guarantees that all consumers with valuations at least $v_{\ell}$ purchase the good, the total social welfare generated by $\ZPC$ is
    \[
    \SW(\ZPC) = \CS(\ZPC) + \PS(\ZPC) = \sum_{i=\ell}^n v_i x^*_i.
    \]
    Combining with the expression for producer surplus, we obtain
    \[
    \CS(\ZPC) = \sum_{i=\ell}^n v_i x^*_i - \UniRev(\bx^*),
    \]
    as claimed.
\end{proof}

\subsection{Social-Welfare-Minimizing \SCHEME} \label{sect:proof-passive-SW}

We now present the social-welfare-minimizing \scheme $\ZPMin$ and characterize the exact value of $\PassiveCSMin(\bx^*, F)$. We begin with the following definition.

\begin{definition}[$\bx^*$-Sub-Feasible Price Set] \label{defn:sub-feasible}
    Given a market $\bx^*$ and an $\bx^*$-feasible price set $F = \{v_\ell, \ldots, v_r\}$, for any $i \in \{\ell, \ldots, r\}$ and $0 \le \eta \le x^*_{i}$, we say that $(F, i, \eta)$ is an $\bx^*$-sub-feasible price set if there exists an \Fvalid \scheme $\Z = \{(\bx_q, p_q)\}_{q \in [Q]}$ such that:
    \begin{enumerate}
        \item For all $q \in [Q]$, $p_q \in \{v_i, v_{i+1}, \ldots, v_r\}$.
        \item $\sum_{q \in [Q]} \bbI[p_q = v_i] \cdot x_{q, i} \le \eta$.
    \end{enumerate}
    Such a \scheme $\Z$ is termed \emph{$(F,i,\eta)$-valid}.
\end{definition}

Intuitively, an $\bx^*$-sub-feasible price set corresponds to an \Fvalid \scheme that only uses prices from $\{v_i, \ldots, v_r\} \subseteq F$ and limits the total market volume assigned to the lowest price $v_i$.

A graphical illustration is shown in the left panel of \cref{fig:proof-SWMin}. In the example, the original price set $F = \{v_2, v_3, v_4\}$ is $\bx^*$-feasible, meaning optimal prices in an \Fvalid \scheme may lie in $F$. However, the set may be reduced to the shaded region $\{v_3, v_4\}$, with a constraint on the total mass allocated to $v_3$.

Given $\bx^*$ and an $\bx^*$-feasible price set $F = \{v_\ell, \ldots, v_r\}$, define:
\begin{equation} \label{eq:passive-SW-min-i0-eta0}
    \begin{aligned}
        i_0 & = \max\{i \in \{\ell, \ldots, r\} : \bx^* \text{ is } \{v_i, \ldots, v_r\}\text{-feasible} \}, \\
        \eta_0 & = \inf\{0 \le \eta \le x^*_{i_0} : (F, i_0, \eta) \text{ is } \bx^*\text{-sub-feasible} \}.
    \end{aligned}
\end{equation}

Since the set $\{ \eta : (F, i_0, \eta) \text{ is } \bx^*\text{-sub-feasible}\}$ is closed, the infimum is attained, and $(F, i_0, \eta_0)$ is $\bx^*$-sub-feasible. This pair identifies the “minimal” feasible subset of $F$ for $\bx^*$ under our constraints.

We now state our main result, which characterizes the bottom-left point in \cref{fig:summary}, showing that minimum consumer surplus and producer surplus are simultaneously achievable.

\begin{theorem} \label{thrm:passive-SW-min}
There exists an \Fvalid \scheme $\ZPMin$ for $\bx^*$ such that
\[
    \CS(\ZPMin) = \PassiveCSMin(\bx^*, F), \quad \PS(\ZPMin) = \UniRev(\bx^*).
\]
Furthermore,
\[
    \PassiveCSMin(\bx^*, F) = \eta_0 v_{i_0} + \sum_{j=i_0+1}^n x^*_j v_j - \UniRev(\bx^*).
\]
\end{theorem}

\paragraph{Proof Sketch of \cref{thrm:passive-SW-min}.}
Let
\begin{equation} \label{eq:tF}
    \tF = \{v_{i_0}, v_{i_0+1}, \ldots, v_r\}
\end{equation}
denote the reduced regulated price set. By construction, $\bx^*$ is $\tF$-feasible.

To implement the idea in \cref{sect:passive-sw-min}, we modify the set definition and construction rule from \cref{sect:passive-ps-max} as follows:
\begin{equation} \label{eq:tB}
    \tB(\bx^*, \tF) = 
    \begin{cases}
        B(\bx^*, \tF) \cup \{v_{i_0}\}, & \text{if } x^*_{i_0} > \eta_0, \\
        B(\bx^*, \tF), & \text{otherwise}.
    \end{cases}
\end{equation}
Then, $\bx^{\equal}$ and $p$ are defined based on the support of $\bx^*$ and $F$:
\begin{enumerate}
    \item \textbf{Case 1:} If $\max(\support(\bx^*) \cap \tF) > v_{i_0}$,
    \begin{equation} \label{eq:construction-PS-max-modified-1}
        \left\{
        \begin{aligned}
            \bx^{\equal} &= \argmax\left\{\mass(\bx) : \bx \in \ERM(\tB(\bx^*, \tF)),\ \bx \le \bx^*,\ x^*_{i_0} - x_{i_0} \ge \eta_0 \right\}, \\
            p &= \min\left(\support(\bx^{\equal}) \cap \{v_{i_0+1}, \ldots, v_r\}\right).
        \end{aligned}
        \right.
    \end{equation}
    The additional constraint $x^*_{i_0} - x_{i_0} \ge \eta_0$ ensures that at least $\eta_0$ mass remains on $v_{i_0}$, allowing us to construct future segments where $v_{i_0}$ is the optimal price.
    \item \textbf{Case 2:} If $\max(\support(\bx^*) \cap \tF) = v_{i_0}$,
    \begin{equation} \label{eq:construction-PS-max-modified-2}
        \left\{
        \begin{aligned}
            \bx^{\equal} &= \argmax\left\{ \mass(\bx) : \bx \in \ERM(\tB(\bx^*, \tF)),\ \bx \le \bx^* \right\}, \\
            p &= v_{i_0}.
        \end{aligned}
        \right.
    \end{equation}
    This corresponds to the standard construction process in \cref{eq:construction-PS-max}.
\end{enumerate}

In \cref{alg:construction}, \cref{line:construction-equation} invokes either \cref{eq:construction-PS-max-modified-1} or \cref{eq:construction-PS-max-modified-2}, depending on the relevant condition. Let $\ZPMin = \{(\bx_q, p_q)\}_{q \in [|F|]}$ denote the output of the algorithm under this modified construction. In \cref{sect:proof-passive-SW-min-feasibility}, we show that $\ZPMin$ is an $(F, i_0, \eta_0)$-valid \scheme of $\bx^*$. Then, in \cref{sect:proof-passive-SW-min-PS,sect:proof-passive-SW-min-CS}, we prove that $\ZPMin$ achieves the desired producer and consumer surplus guarantees.
\subsubsection{Proof of \cref{thrm:passive-SW-min}: Feasibility} \label{sect:proof-passive-SW-min-feasibility}
We first show that the output $\ZPMin$ is $(F, i_0, \eta_0)$-valid. Similar to the producer-surplus-maximizing construction in \cref{sect:proof-passive-PS}, we aim to ensure that the \scheme~$Z$ constructed by \cref{alg:construction} is $(F, i_0, \eta_0)$-valid when the while loop terminates, even if $\bx^{\remain} \ne \boldsymbol{0}$.

To this end, we introduce two auxiliary modifications—later shown to be unnecessary—that enforce the $(F, i_0, \eta_0)$-validity explicitly:

\begin{enumerate}
    \item To ensure that at most $\eta_0$ mass of $v_{i_0}$ is used as the optimal price, we modify \cref{eq:construction-PS-max-modified-2} to:
    \begin{equation} \label{eq:construction-PS-max-modified-2-new}
        \left\{
        \begin{aligned}
            \bx^{\equal} &= \argmax\left\{ \mass(\bx) : \bx \in \ERM(\tB(\bx^*, \tF)),\ \bx \le \bx^*,\ \sum_{(\by, p) \in \Z,\ p = v_{i_0}} y_{i_0} + x_{i_0} \le \eta_0 \right\}, \\
            p &= v_{i_0},
        \end{aligned}
        \right.
    \end{equation}
    where $\Z$ denotes the partial \scheme constructed thus far.

    \item We modify the while-loop condition in \cref{line:while} of \cref{alg:construction} from $\support(\bx^*) \cap F \ne \emptyset$ to:
    \begin{equation} \label{eq:while-condition}
        \support(\bx^*) \cap \tF \ne \emptyset \quad \text{and} \quad \sum_{(\by, p) \in \Z,\ p = v_{i_0}} y_{i_0} < \eta_0.
    \end{equation}
    This ensures the loop exits once no further segment can be added without violating $(F, i_0, \eta_0)$-validity.
\end{enumerate}

We now define $\ZPMin' = \{(\bx_q, p_q)\}_{q \in [|F|]}$ as the output of the algorithm with \cref{line:construction-equation} set according to \cref{eq:construction-PS-max-modified-1,eq:construction-PS-max-modified-2-new}. By construction, $\ZPMin'$ is $(F, i_0, \eta_0)$-valid. Similar to \cref{thrm:passive-ps-ps}, we could get the following lemma.
\begin{lemma} \label{lemma:passive-SW-min-ps}
    For any $(F, i_0, \eta_0)$-valid \scheme~$\Z$, we have
    \[
    \PS(\ZPMin') \ge \PS(\Z).
    \]
\end{lemma}
\begin{proof}
    Note that $\tB(\bx^*, \tF)$ in \cref{eq:tB} differs from $B(\bx^*, \tF)$ in \cref{eq:S} only in the item $v_{i_0}$. The result then follows directly by applying the same inductive argument used in the proof of \cref{thrm:passive-ps-ps} (see \cref{sec:proof-passive-PS-PS} for details).
\end{proof}

\begin{proposition}
    The output $\ZPMin'$ is a $(F, i_0, \eta_0)$-valid \scheme of $\bx^*$. Moreover, we have $\ZPMin' = \ZPMin$, implying that the additional modifications in \cref{eq:construction-PS-max-modified-2-new,eq:while-condition} are unnecessary.
\end{proposition}
\begin{proof}
    Let $\bx^{\remain}$ be the residual market after executing \cref{alg:construction} using \cref{eq:construction-PS-max-modified-1,eq:construction-PS-max-modified-2-new} and the modified while-loop condition in \cref{eq:while-condition}. Suppose, for contradiction, that $\bx^{\remain} > \boldsymbol{0}$. Let $v \in \support(\bx^{\remain})$. We show that $v \in \optPrice(\bx)$ for any $(\bx, p) \in \ZPMin'$ with $p \ge v_{i_0}$.
    
    \begin{enumerate}
        \item If $v \ne v_{i_0}$, then $v \in \support(\bx^*)$ throughout the while-loop. By the construction of $\tB(\bx^*, \tF)$ in \cref{eq:tB}, it follows that $v \in \support(\bx^{\equal})$ and hence $v \in \optPrice(\bx)$ for any $(\bx, p) \in \ZPMin'$ with $p \ge v_{i_0}$.

        \item If $v = v_{i_0}$, then during the while loop, we must have had $x^*_{i_0} > \eta_0$ whenever $\max(\support(\bx^*) \cap \tF) > v_{i_0}$. Therefore, $v \in \optPrice(\bx)$ for any $(\bx, p) \in \ZPMin'$ with $p > v_{i_0}$. Since $v = v_{i_0}$, it follows that $v \in \optPrice(\bx)$ also for $p = v_{i_0}$, and thus for all $p \ge v_{i_0}$.
    \end{enumerate}
    In addition, by construction, it holds that $\bx = \boldsymbol{0}$ for any $(\bx, p) \in \ZPMin'$ with $p < v_{i_0}$. Now, using the same argument as in the proof of \cref{lemma:passive-PS-x-equals-0}, let $\Z_1$ be any $(F, i_0, \eta_0)$-valid scheme for $\bx^*$. Then:
    \begin{align*}
        R_{\bx^*}(v) &= R_{\bx^{\remain}}(v) + \sum_{(\by, p) \in \ZPMin'} R_{\by}(v) \\
        &= R_{\bx^{\remain}}(v) + \sum_{(\by, p) \in \ZPMin'} R_{\by}(p) \\
        &= R_{\bx^{\remain}}(v) + \PS(\ZPMin') \\
        &\ge R_{\bx^{\remain}}(v) + \PS(\Z_1) \tag{by \cref{lemma:passive-SW-min-ps}} \\
        &> \PS(\Z_1) \tag{since $v \in \support(\bx^{\remain})$ and $R_{\bx^{\remain}}(v) > 0$} \\
        &\ge R_{\bx^*}(v).
    \end{align*}
    This leads to a contradiction: $R_{\bx^*}(v) > R_{\bx^*}(v)$. Therefore, we must have $\bx^{\remain} = \boldsymbol{0}$, and hence $\ZPMin'$ is a $(F, i_0, \eta_0)$-valid \scheme for $\bx^*$.

    Finally, since $x^{\remain}_{i_0} = 0$, the additional constraint $\sum_{(\by, p) \in \Z,\ p = v_{i_0}} y_{i_0} + x_{i_0} \le \eta_0$ in \cref{eq:construction-PS-max-modified-2-new} is vacuously satisfied and thus unnecessary. Moreover, because $\bx^{\remain} = \boldsymbol{0}$ implies $\support(\bx^*) \cap F = \emptyset$, the modification to the while-loop condition in \cref{eq:while-condition} is also unnecessary. Hence, $\ZPMin = \ZPMin'$, completing the proof.
\end{proof}

\subsubsection{Proof of \cref{thrm:passive-SW-min}: Producer Surplus Minimization} \label{sect:proof-passive-SW-min-PS}
As in the proof of \cref{thrm:passive-CS} in \cref{sect:passive-cs-max}, we aim to show that there exists an optimal price $p^*$ of $\bx^*$ such that $p^* \in \optPrice(\bx)$ for every $(\bx, p) \in \ZPMin$. We establish this through a case analysis, formalized in the following lemmas.

\begin{lemma} \label{lemma:passive-SW-min-pstar-1}
    When $|\tF| = 1$ (i.e., $\tF = \{v_{i_0}\}$), we have $v_{i_0} \in \optPrice(\bx)$ for every $(\bx, p) \in \ZPMin$.
\end{lemma}
\begin{proof}
    By definition, $\tF$ is $\bx^*$-feasible. Since $\tF$ contains only one price, all segments in $\ZPMin$ must be zero except for a single segment of the form $(\bx^*, v_{i_0})$. The result then follows immediately.
\end{proof}

\begin{lemma} \label{lemma:passive-SW-min-pstar-2}
    Suppose $|\tF| > 1$ and $\ZPMin = \{(\bx_q, p_q)\}_{q \in [|F|]}$ with $p_q = v_{\ell + q - 1}$. Let $q_0 = i_0 + 2 - \ell$. If $v_{i_0} \in \optPrice(\bx_{q_0})$, then $v_{i_0} \in \optPrice(\bx_q)$ for all $q \in [|F|]$.
\end{lemma}
\begin{proof}
    By the construction of $\tB(\bx^*, \tF)$ in \cref{eq:tB,eq:construction-PS-max-modified-1,eq:construction-PS-max-modified-2}, the fact that $v_{i_0} \in \optPrice(\bx_{q_0})$ implies that $v_{i_0} \in \optPrice(\bx_q)$ for all $q \ge q_0$. Since $p_{q_0 - 1} = v_{i_0}$ by the definition of $q_0$, we indeed have $v_{i_0} \in \optPrice(\bx_{q_0 - 1})$. Moreover, because $\ZPMin$ is $(F, i_0, \eta_0)$-valid, we know that $\bx_q = \boldsymbol{0}$ for all $q < q_0 - 1$. In these cases, $v_{i_0} \in \optPrice(\bx_q)$ holds vacuously. Hence, $v_{i_0} \in \optPrice(\bx_q)$ for all $q \in [|F|]$. This completes the proof.
\end{proof}

\begin{lemma} \label{lemma:passive-SW-min-pstar-3}
    Suppose $|\tF| > 1$ and let $\ZPMin = \{(\bx_q, p_q)\}_{q \in [|F|]}$ with $p_q = v_{\ell + q - 1}$. Define $q_0 = i_0 + 2 - \ell$. If $v_{i_0} \notin \optPrice(\bx_{q_0})$, then there exists some $p^* \in V$ such that $p^* \in \optPrice(\bx_q)$ for all $q \in [|F|]$.
\end{lemma}
\begin{proof}
    By construction, we have $v_{i_0} \in \optPrice(\bx_{q_0 - 1})$. We now show that the optimal price of $\bx_{q_0 -1}$ is not unique.

    Suppose, for contradiction, that $\optPrice(\bx_{q_0 - 1}) = \{p_{q_0 - 1}\} = \{v_{i_0}\}$ is unique. We will then construct an alternative \scheme that is still \Fvalid but leads to a strictly smaller $\eta_0$—contradicting the minimality of $\eta_0$ in \cref{eq:passive-SW-min-i0-eta0}. Consider two cases:

    \textbf{Case 1:} $\support(\bx_{q_0}) \cap \{v_1, \dots, v_{i_0 - 1}\} = \emptyset$.  
    Define alternative segments:
    \[
    x'_{q_0-1, i} = \begin{cases}
        x_{q_0-1, i} & \text{if } i \ne i_0, \\
        x_{q_0-1, i} - \varepsilon & \text{if } i = i_0,
    \end{cases}
    \quad
    x'_{q_0, i} = \begin{cases}
        x_{q_0, i} & \text{if } i \ne i_0, \\
        x_{q_0, i} + \varepsilon & \text{if } i = i_0.
    \end{cases}
    \]

    \textbf{Case 2:} $\support(\bx_{q_0}) \cap \{v_1, \dots, v_{i_0 - 1}\} \ne \emptyset$.  
    Let $i_1$ be the maximal index in this intersection. Define:
    \[
    x'_{q_0-1, i} = \begin{cases}
        x_{q_0-1, i} & \text{if } i \notin \{i_0, i_1\}, \\
        x_{q_0-1, i} - \varepsilon & \text{if } i = i_0, \\
        x_{q_0-1, i} + \varepsilon & \text{if } i = i_1,
    \end{cases}
    \quad
    x'_{q_0, i} = \begin{cases}
        x_{q_0, i} & \text{if } i \notin \{i_0, i_1\}, \\
        x_{q_0, i} + \varepsilon & \text{if } i = i_0, \\
        x_{q_0, i} - \varepsilon & \text{if } i = i_1.
    \end{cases}
    \]

    In both cases, for sufficiently small $\varepsilon > 0$, we have $p_{q_0 - 1} \in \optPrice(\bx'_{q_0-1})$ and $p_{q_0} \in \optPrice(\bx'_{q_0})$. Thus, replacing $(\bx_{q_0-1}, p_{q_0-1})$ and $(\bx_{q_0}, p_{q_0})$ with $(\bx'_{q_0-1}, p_{q_0-1})$ and $(\bx'_{q_0}, p_{q_0})$ yields a new \Fvalid \scheme with strictly smaller $\eta_0$, contradicting minimality.

    Therefore, the optimal price of $\bx_{q_0 - 1}$ is not unique: there exists $p \ne v_{i_0}$ such that $p \in \optPrice(\bx_{q_0 - 1})$. Moreover, by construction, $x_{q_0 - 1, i} = 0$ for all $v_i \in \{v_{i_0+1}, \dots, v_r\}$, so $p \notin \{v_{i_0+1}, \dots, v_r\}$. Hence, $p \in V \setminus \tF$.

    From the construction in \cref{eq:tB,eq:construction-PS-max-modified-1,eq:construction-PS-max-modified-2}, it follows that $p \in \optPrice(\bx_q)$ for all $q \ge q_0$. Since $\ZPMin$ is $(F, i_0, \eta_0)$-valid, we also have $\bx_q = \boldsymbol{0}$ for all $q < q_0 - 1$, and thus $p \in \optPrice(\bx_q)$ vacuously.

    Hence, $p \in \optPrice(\bx_q)$ for all $q \in [|F|]$. This completes the proof.
\end{proof}

Based on the lemmas above, we conclude that the producer surplus achieved by $\ZPMin$ equals the maximum uniform-price revenue of the market $\bx^*$, i.e., $\PS(\ZPMin) = \UniRev(\bx^*)$.
\begin{proposition} \label{prop:passive-SW-min-PS}
    It holds that
    \[
    \PS(\ZPMin) = \UniRev(\bx^*).
    \]
\end{proposition}
\begin{proof}
    Let $\ZPMin = \{(\bx_q, p_q)\}_{q \in [|F|]}$ with $p_q = v_{\ell + q - 1}$. By \cref{lemma:passive-SW-min-pstar-1,lemma:passive-SW-min-pstar-2,lemma:passive-SW-min-pstar-3}, there exists a price $p^* \in \optPrice(\bx^*)$ such that $p^* \in \optPrice(\bx_q)$ for all $q \in [|F|]$. Therefore,
    \[
    \PS(\ZPMin) = \sum_{q=1}^{|F|} R_{\bx_q}(p_q) = \sum_{q=1}^{|F|} R_{\bx_q}(p^*) = R_{\bx^*}(p^*) = \UniRev(\bx^*),
    \]
    where the last equality follows from the definition of $\UniRev(\bx^*)$ as the maximum revenue of $\bx^*$ under uniform pricing. This completes the proof.
\end{proof}

\subsubsection{Proof of \cref{thrm:passive-SW-min}: Consumer Surplus Minimization} \label{sect:proof-passive-SW-min-CS}
Let $\ZPP$ denote the output of \cref{alg:construction} when \cref{eq:construction-PS-max} is used in \cref{line:construction-equation}. By \cref{thrm:passive-PS}, we know that $\ZPP$ is \Fvalid on $\bx^*$. We now show that $\ZPP$ and $\ZPMin$ share a similar structure.

\begin{lemma} \label{lemma:passive-SW-min-x-equals-y-large}
    Suppose $\ZPMin = \{(\bx_q, p_q)\}_{q \in [|F|]}$ and $\ZPP = \{(\by_q, p_q)\}_{q \in [|F|]}$ with $p_q = v_{\ell + q - 1}$. Then
    \[
        x_{q,i} = y_{q,i}, \quad \forall\, q \in [|F|] \text{ such that } p_q > v_{i_0},\; \forall\, i \in \{i_0 + 1, i_0 + 2, \dots, n\}.
    \]
\end{lemma}

\begin{proof}
    Define
    \[
    C = \{v_{i_0+1}, v_{i_0+2}, \dots, v_n\}.
    \]
    Let $\{\bx^{\equal}_j\}_{j \in [m_1]}$ be the equal-revenue markets constructed by \cref{alg:construction} using \cref{eq:construction-PS-max-modified-1,eq:construction-PS-max-modified-2} in \cref{line:construction-equation}. Similarly, let $\{\by^{\equal}_j\}_{j \in [m_2]}$ be the markets constructed using \cref{eq:construction-PS-max}.

    We group $\{\bx^{\equal}_j\}_{j \in [m_1]}$ according to their intersection with $C$. Let $\{\bx'_j\}_{j \in [m'_1]}$ be the resulting grouped markets, defined iteratively as follows:

    \begin{enumerate}
        \item Let $\bx'_1 = \sum_{j \in [m_1]} \bbI[\support(\bx^{\equal}_j) \cap C = \support(\bx^{\equal}_1) \cap C] \cdot \bx^{\equal}_j$.
        \item Let $t_1$ be the first index $i$ such that $\support(\bx^{\equal}_i) \cap C \ne \support(\bx^{\equal}_1) \cap C$. Define
        \[
        \bx'_2 = \sum_{j \in [m_1]} \bbI[\support(\bx^{\equal}_j) \cap C = \support(\bx^{\equal}_{t_1}) \cap C] \cdot \bx^{\equal}_j.
        \]
        \item Let $t_2$ be the first index $i > t_1$ such that $\support(\bx^{\equal}_i) \cap C \ne \support(\bx^{\equal}_{t_1}) \cap C$, and define $\bx'_3$ analogously.
        \item Continue this process until all $\bx^{\equal}_j$ are grouped into $m'_1$ disjoint aggregated markets, each corresponding to a unique subset of $C$.
    \end{enumerate}

    Similarly, we can group $\{\by^{\equal}_j\}_{j \in [m_2]}$ based on $\support(\by^{\equal}_j) \cap C$ to obtain $\{\by'_j\}_{j \in [m'_2]}$.

    By the construction of $B(\bx^*, F)$ in \cref{eq:S} and $\tB(\bx^*, \tF)$ in \cref{eq:tB}, we have
    \[
    B(\bx^*, F) \cap C = \tB(\bx^*, \tF) \cap C.
    \]
    Moreover, by the definition of equal-revenue markets in \cref{eq:equal-revenue}, we observe that for any subsets $H_1, H_2 \subseteq \{v_1, v_2, \dots, v_{i_0}\}$ and any subset $J \subseteq \{v_{i_0+1}, v_{i_0+2}, \dots, v_n\}$, the mass of $\bx^{H_1 \cup J}$ (defined in \cref{eq:equal-revenue}) on $C$ is proportional to the mass of $\bx^{H_2 \cup J}$ on $C$. Formally,
    \[
    \forall\, i \in \{i_0 + 1, i_0 + 2, \dots, n\}, \quad \frac{x^{H_1 \cup J}_i}{x^{H_2 \cup J}_i} \text{ is constant}.
    \]
    Since both $\{\bx'_j\}_{j \in [m'_1]}$ and $\{\by'_j\}_{j \in [m'_2]}$ form segmentations of the same market $\bx^*$, it follows that
    \[
    m'_1 = m'_2 \quad \text{and} \quad x'_{j,i} = y'_{j,i}, \quad \forall\, j \in [m'_1],\ i \in \{i_0 + 1, i_0 + 2, \dots, n\}.
    \]
    As a result, for any $q \in [|F|]$ such that $p_q > v_{i_0}$ and for any $i \in \{i_0 + 1, \dots, n\}$, we have
    \[
    x_{q,i} = \sum_{j \in [m'_1]} \bbI\left[\min(\support(\bx'_j) \cap C) = v_i\right] \cdot x'_{j,i}
    = \sum_{j \in [m'_1]} \bbI\left[\min(\support(\by'_j) \cap C) = v_i\right] \cdot y'_{j,i}
    = y_{q,i}.
    \]
    This completes the proof.
\end{proof}

\begin{lemma} \label{lemma:passive-SW-min-x-equals-y-small}
    Suppose $\ZPMin = \{(\bx_q, p_q)\}_{q \in [|F|]}$ and $\ZPP = \{(\by_q, p_q)\}_{q \in [|F|]}$ with $p_q = v_{\ell + q - 1}$. Then
    \[
        x_{q,i} = y_{q,i} = 0, \quad \forall\, q \in [|F|] \text{ such that } p_q < v_{i_0},\; \forall\, i \in \{i_0 + 1, i_0 + 2, \dots, n\}.
    \]
\end{lemma}
\begin{proof}
    For $\ZPMin$, by \cref{lemma:passive-SW-min-ps}, we have that $\bx_q = \boldsymbol{0}$ for all $q \in [|F|]$ such that $p_q < v_{i_0}$. Therefore, $x_{q,i} = 0$ for all $i \in \{i_0 + 1, i_0 + 2, \dots, n\}$.

    The same conclusion holds for $\ZPP$, since $\bx^*$ is $\tF$-feasible with $\tF = \{v_{i_0}, v_{i_0+1}, \dots, v_n\}$. Thus, the construction ensures that for any $q$ with $p_q < v_{i_0}$, we have $\by_q = \boldsymbol{0}$, and hence $y_{q,i} = 0$ for all $i > i_0$. This completes the proof.
\end{proof}

\begin{lemma} \label{lemma:passive-SW-min-x-equals-y-middle}
    Suppose $\ZPMin = \{(\bx_q, p_q)\}_{q \in [|F|]}$ and $\ZPP = \{(\by_q, p_q)\}_{q \in [|F|]}$ with $p_q = v_{\ell + q - 1}$. Then
    \[
        x_{q,i} = y_{q,i}, \quad \forall\, q \in [|F|] \text{ such that } p_q = v_{i_0},\; \forall\, i \in \{i_0 + 1, i_0 + 2, \dots, n\}.
    \]
\end{lemma}
\begin{proof}
    Let $q_0$ be the index such that $p_{q_0} = v_{i_0}$. Then, by \cref{lemma:passive-SW-min-x-equals-y-large,lemma:passive-SW-min-x-equals-y-small}, it holds that for all $i \in \{i_0+1, \dots, n\}$,
    \[
    x_{q_0,i} = x^*_i - \sum_{\substack{q \in [|F|] \\ p_q < v_{i_0}}} x_{q,i} - \sum_{\substack{q \in [|F|] \\ p_q > v_{i_0}}} x_{q,i}
    = x^*_i - \sum_{\substack{q \in [|F|] \\ p_q < v_{i_0}}} y_{q,i} - \sum_{\substack{q \in [|F|] \\ p_q > v_{i_0}}} y_{q,i}
    = y_{q_0,i}.
    \]
    This completes the proof.
\end{proof}

Now we can show that the consumer surplus of $\ZPMin$ is exactly the same as that of $\ZPP$, implying that $\CS(\ZPMin) = \PassiveCSMin(\bx^*, F)$.

\begin{proposition}
    Suppose $\ZPMin = \{(\bx_q, p_q)\}_{q \in [|F|]}$ and $\ZPP = \{(\by_q, p_q)\}_{q \in [|F|]}$ with $p_q = v_{\ell + q - 1}$. Then
    \[
    \CS(\ZPMin) = \CS(\ZPP) = \PassiveCSMin(\bx^*, F).
    \]
    In addition,
    \[
    \PassiveCSMin(\bx^*, F) = \eta_0 v_{i_0} + \sum_{j=i_0+1}^n x^*_j v_j - \UniRev(\bx^*).
    \]
\end{proposition}

\begin{proof}
    Combining \cref{lemma:passive-SW-min-x-equals-y-large,lemma:passive-SW-min-x-equals-y-small,lemma:passive-SW-min-x-equals-y-middle}, we have
    \begin{equation} \label{eq:passive-SW-min-x-equals-y}
        x_{q,i} = y_{q,i}, \quad \forall\, q \in [|F|],\; \forall\, i \in \{i_0 + 1, i_0 + 2, \dots, n\}.
    \end{equation}

    As a result, by the construction of $\ZPMin$ and $\ZPP$, only consumers with valuations at least $v_{r+1}$ can obtain positive surplus. Therefore,
    \[
    \begin{aligned}
        \CS(\ZPMin) &= \sum_{q=1}^{|F|} \sum_{i=\ell}^n (v_i - v_{\ell+q-1}) x_{q,i}
        = \sum_{q=1}^{|F|} \sum_{i=r+1}^n (v_i - v_{\ell+q-1}) x_{q,i}, \\
        \CS(\ZPP) &= \sum_{q=1}^{|F|} \sum_{i=\ell}^n (v_i - v_{\ell+q-1}) y_{q,i}
        = \sum_{q=1}^{|F|} \sum_{i=r+1}^n (v_i - v_{\ell+q-1}) y_{q,i}.
    \end{aligned}
    \]
    The equality $\CS(\ZPMin) = \CS(\ZPP)$ follows directly from \cref{eq:passive-SW-min-x-equals-y}. Since $\CS(\ZPP) = \PassiveCSMin(\bx^*, F)$ by \cref{thrm:passive-PS}, it holds that $\CS(\ZPMin) = \PassiveCSMin(\bx^*, F)$. Combining with \cref{prop:passive-SW-min-PS} and noting that $\CS(\ZPMin) + \PS(\ZPMin) = \eta_0 v_{i_0} + \sum_{j=i_0+1}^n x^*_j v_j$ by construction, we obtain
    \[
    \PassiveCSMin(\bx^*, F) = \eta_0 v_{i_0} + \sum_{j=i_0+1}^n x^*_j v_j - \UniRev(\bx^*),
    \]
    which completes the proof.
\end{proof}

\section{Active Intermediary Model And Results} \label{sec:active-inter}
\subsection{Model and Preliminaries}

In the active intermediary setting, the seller is constrained to post prices within the regulated set $F = \{v_{\ell}, v_{\ell+1}, \dots, v_r\}$. Thus, the intermediary only needs to ensure that the instructed price is one of the revenue-maximizing prices \emph{within} $F$, making the feasibility condition weaker than the incentive compatibility constraint in \cref{def:inter_IC}.

We use the notation
\[
\fopt(\bx) = \arg\max_{v \in F} R_\bx(v)
\]
to denote the set of revenue-maximizing prices for market $\bx$ restricted to the feasible set $F$.

\begin{definition}[$F$-instructed \SCHEME]
    For a contiguous price set $F \subseteq V$, we say that a \scheme $\Z = \{(\bx_q, p_q)\}_{q \in [Q]}$ is $F$-instructed if for all $q \in [Q]$, $p_q \in \fopt(\bx_q)$.
\end{definition}

The $F$-instructed property imposes a weaker requirement than $F$-validity (\cref{def:inter_IC}). In an $F$-instructed scheme, the price $p_q$ is optimal \emph{within} $F$, whereas in an $F$-valid scheme, the price must be optimal over the entire value set $V$ \emph{and} fall inside $F$. As a result, any market segmentation admits at least one set of instructed prices that make the scheme $F$-instructed, but not necessarily $F$-valid.

Accordingly, the central question in the active intermediary scenario is:
\begin{quote}
    \emph{Among all $F$-instructed \schemes for a given market $\bx^*$, what is the achievable region of consumer surplus and producer surplus?}
\end{quote}

\subsection{Results}
Since all $p \in \fopt(\bx^*, F)$ yield the same revenue, we define
\[
R^F_{\uni}(\bx^*) = R_{\bx^*}(v^*)
\]
for any $v^* \in \fopt(\bx^*, F)$. Given that all instructed prices in an $F$-instructed scheme lie within $F$, buyers with value $v_i$ exceeding $\max(F)$ must receive at least $v_i - \max(F)$ in surplus. We denote this lower bound by
\[
\ActiveCSMin(\bx^*, F) = \sum_{i=r+1}^n x_i^* \cdot (v_i - v_r).
\]
Similar to the passive intermediary model, the set of achievable (consumer surplus, producer surplus) pairs is characterized by three inequalities:

\begin{theorem} \label{thm:RSeller_base_bounds}
    Let $\bx^*$ be a market and $F = \{v_\ell, v_{\ell+1}, \dots, v_r\} \subseteq V$ be a contiguous price set. Then any achievable pair $(x, y)$ under an $F$-instructed \scheme satisfies:
    \begin{enumerate}
        \item $x + y \le \PassiveSWMax(\bx^*, F)$,
        \item $x \ge \ActiveCSMin(\bx^*, F)$,
        \item $y \ge R^F_{\uni}(\bx^*)$.
    \end{enumerate}
\end{theorem}

In \cref{sec:Rseller_ps_opt}, we construct a seller-optimal scheme that achieves both $\ActiveCSMin(\bx^*, F)$ in consumer surplus and $\PassiveSWMax(\bx^*, F)$ in social welfare, thereby maximizing producer surplus. In \cref{sec:Rseller_cs_opt}, we construct a buyer-optimal scheme that achieves both $R^F_{\uni}(\bx^*)$ in producer surplus and $\PassiveSWMax(\bx^*, F)$ in social welfare, thus maximizing consumer surplus. Finally, \cref{sec:Rseller_min} presents a scheme that simultaneously achieves the minima of both consumer and producer surplus.

\subsubsection{Producer-Surplus-Maximizing \SCHEME} \label{sec:Rseller_ps_opt}
\begin{proposition}
    There exists an $F$-instructed \scheme $\ZAP$ on $\bx^*$ such that
    \[
    \PS(\ZAP) = \sum_{i=\ell}^r x^*_i \cdot v_i + v_r \cdot \sum_{i=r+1}^n x^*_i, \quad \CS(\ZAP) = \ActiveCSMin(\bx^*, F),
    \]
    and hence
    \[
    \SW(\ZAP) = \PS(\ZAP) + \CS(\ZAP) = \PassiveSWMax(\bx^*, F).
    \]
\end{proposition}

\begin{proof}
    Construct $\ZAP = \{(\bx_q, p_q)\}_{q \in [|F|]}$ with $p_q = v_{\ell + q - 1}$ as follows. Define the market segments:
    \[
    x_{1, i} = \begin{cases}
        x^*_i & \text{if } i \le \ell, \\
        0 & \text{otherwise},
    \end{cases} \qquad
    x_{r - \ell + 1, i} = \begin{cases}
        x^*_i & \text{if } i \ge r, \\
        0 & \text{otherwise},
    \end{cases}
    \]
    and for $q = 2, \ldots, r - \ell$,
    \[
    x_{q,i} = \begin{cases}
        x^*_i & \text{if } i = \ell + q - 1, \\
        0 & \text{otherwise}.
    \end{cases}
    \]
    It is straightforward to verify that $p_q \in \fopt(\bx_q)$ for all $q$, so $\ZAP$ is indeed $F$-instructed. Compute the surpluses:
    \[
    \begin{aligned}
        \ps(\bx_1, p_1) &= x^*_{\ell} \cdot v_{\ell}, &\quad \cs(\bx_1, p_1) &= 0, \\
        \ps(\bx_{r - \ell + 1}, p_{r - \ell + 1}) &= v_r \cdot \sum_{i=r}^n x^*_i, &\quad \cs(\bx_{r - \ell + 1}, p_{r - \ell + 1}) &= \sum_{i=r+1}^n x^*_i \cdot (v_i - v_r), \\
        \ps(\bx_q, p_q) &= v_{\ell + q - 1} \cdot x^*_{\ell + q - 1}, &\quad \cs(\bx_q, p_q) &= 0 \quad \text{for } q = 2, \ldots, r - \ell.
    \end{aligned}
    \]
    Summing over all segments:
    \[
    \begin{aligned}
        \PS(\ZAP) &= \sum_{i=\ell}^r x^*_i \cdot v_i + v_r \cdot \sum_{i=r+1}^n x^*_i, \\
        \CS(\ZAP) &= \sum_{i=r+1}^n x^*_i \cdot (v_i - v_r).
    \end{aligned}
    \]
    This completes the proof.
\end{proof}

\subsubsection{Consumer-Surplus-Maximizing \SCHEME} \label{sec:Rseller_cs_opt}
The consumer-surplus-maximizing \scheme and the social-welfare-minimizing \scheme share a common segmentation structure, differing only in their instructed prices. Given an aggregate market $\bx^*$ and a contiguous price set $F = \{v_{\ell}, \dots, v_r\} \subseteq V$, we construct a modified market $\tbx^* = (\tx^*_1, \tx^*_2, \ldots, \tx^*_n)$ over the value set $V$ as follows:
\[
\tx^*_i = \begin{cases}
    0 & \text{if } i < \ell, \\
    x^*_i & \text{if } \ell \le i < r, \\
    G_{\bx^*}(v_r) & \text{if } i = r, \\
    0 & \text{if } i > r,
\end{cases}
\qquad \forall i \in [n],
\]
where $G_{\bx}(v) = \sum_{i: v_i \ge v} x_i$ denotes the right-tail cumulative mass. That is, we truncate all mass below $F$ and relocate all mass above $F$ to the maximum value in $F$.

We then apply the decomposition method of \citet{bergemann2015limits} (see \cref{eq:construction-bbm}) to $\tbx^*$, yielding a collection $\tX = \{\tbx_q\}_{q \in [Q]}$ of equal-revenue segments. Based on these, we define a segmentation $\{\bx_q\}_{q \in [Q]}$ of the original market $\bx^*$ as:
\begin{equation} \label{eq:active-segmentation}
    x_{1,i} = \begin{cases}
        x^*_i & \text{if } i < \ell, \\
        \tx_{1,i} & \text{if } \ell \le i < r, \\
        \tx_{1,r} \cdot \frac{x^*_i}{G_{\bx^*}(v_r)} & \text{if } i \ge r;
    \end{cases} \qquad
    x_{q,i} = \begin{cases}
        \tx_{q,i} & \text{if } i < r, \\
        \tx_{q,r} \cdot \frac{x^*_i}{G_{\bx^*}(v_r)} & \text{if } i \ge r
    \end{cases} \quad \forall q > 1, \; i \in [n].
\end{equation}
By construction, the segmentation $\{\bx_q\}_{q \in [Q]}$ satisfies the following:
\begin{property} \label{prpt:active}
    For all $q \in [Q]$, we have $\fopt(\bx^*) \subseteq \bx_q$. Moreover, for all $v_i, v_j \in \support(\bx_q) \cap F$, $R_{\bx_q}(v_i) = R_{\bx_q}(v_j)$.
\end{property}

This structure enables the construction of the consumer-surplus-maximizing \scheme $\ZAC$ presented here and the social-welfare-minimizing \scheme $\ZAMin$ described in \cref{sec:Rseller_min}.

\begin{proposition} \label{prop:active-cs}
    There exists an $F$-instructed \scheme $\ZAC$ on $\bx^*$ such that
    \[
    \PS(\ZAC) = R^F_{\uni}(\bx^*), \quad \CS(\ZAC) = \PassiveSWMax(\bx^*, F) - R^F_{\uni}(\bx^*).
    \]
\end{proposition}

\begin{proof}
    Construct $\ZAC = \{(\bx_q, p_q)\}_{q \in [Q]}$ with $\bx_q$ given in \cref{eq:active-segmentation} and
    \[
        p_q = \min(\support(\bx_q) \cap F).
    \]
    By \cref{prpt:active}, we have $p_q \in \fopt(\bx_q)$, and thus $\ZAC$ is an $F$-instructed \scheme on $\bx^*$. Let $p^* \in \fopt(\bx^*)$ be any feasible uniform price. By \cref{prpt:active}, it holds that
    \[
    \PS(\ZAC) = \sum_{q \in [Q]} \ps(\bx_q, p_q) = \sum_{q \in [Q]} \ps(\bx_q, p^*) = \ps(\bx^*, p^*) = R^F_{\uni}(\bx^*).
    \]

    Furthermore, by the choice of $p_q$, all consumers with values at least $v_\ell$ will make purchases, implying that
    \[
    \PS(\ZAC) + \CS(\ZAC) = \PassiveSWMax(\bx^*, F),
    \]
    which completes the proof.
\end{proof}

\subsubsection{Social-Welfare-Minimizing \SCHEME} \label{sec:Rseller_min}
\begin{proposition}
    There exists an $F$-instructed \scheme $\ZAMin$ on $\bx^*$ such that
    \[
    \PS(\ZAMin) = R^F_{\uni}(\bx^*), \quad \CS(\ZAMin) = \ActiveCSMin(\bx^*, F).
    \]
\end{proposition}

\begin{proof}
    Construct $\ZAMin = \{(\bx_q, p_q)\}_{q \in [Q]}$ with $\bx_q$ given in \cref{eq:active-segmentation} and
    \[
        p_q = \max(\support(\bx_q) \cap F).
    \]

    For the producer surplus, similar to \cref{prop:active-cs}, let $p^* \in \fopt(\bx^*)$ be any feasible uniform price. By \cref{prpt:active}, it holds that
    \[
    \PS(\ZAMin) = \sum_{q \in [Q]} \ps(\bx_q, p_q) = \sum_{q \in [Q]} \ps(\bx_q, p^*) = \ps(\bx^*, p^*) = R^F_{\uni}(\bx^*).
    \]

    For the consumer surplus, we calculate $\cs(\bx_q, p_q)$ based on the support of $\bx_q$:
    \begin{enumerate}
        \item \textbf{Case 1:} $\max(\support(\bx_q) \cap F) < v_r$. By construction (\cref{eq:active-segmentation}), $x_{q,i} = 0$ for all $i > r$, so $\cs(\bx_q, p_q) = 0$.
        \item \textbf{Case 2:} $\max(\support(\bx_q) \cap F) = v_r$. Then $p_q = v_r$, and
        \[
        \cs(\bx_q, p_q) = \sum_{i = r+1}^n x_{q,i}(v_i - v_r).
        \]
    \end{enumerate}
    Summing over all $q \in [Q]$, we get:
    \begin{align*}
    \CS(\ZAMin) &= \sum_{q \in [Q]} \cs(\bx_q, p_q) = \sum_{q \in [Q]} \bbI[\max(\support(\bx_q) \cap F) = v_r] \sum_{i = r+1}^n x_{q,i}(v_i - v_r) \\
    &= \sum_{q \in [Q]} \sum_{i = r+1}^n x_{q,i}(v_i - v_r) \tag{since $x_{q,i} = 0$ for $i > r$ in Case 1} \\
    &= \sum_{i=r+1}^n x^*_i (v_i - v_r) = \ActiveCSMin(\bx^*, F).
    \end{align*}
    This completes the proof.
\end{proof}

\section{More Discussions on the Passive Intermediary Model} \label{sect:discussion-passive-model}
\subsection{Feasibility with the Optimal Uniform Price Excluded} \label{sect:justification-feasibility}
In this section, we analyze the cases when the regulated set $F$ does not contain the optimal price but remains $\bx^*$-feasible. 

A natural question on these cases is whether price sets that deviate significantly from the uniform price can be feasible. In contrast to the intuition that significant deviations make regulation infeasible, there are examples where large deviated price sets are feasible. The following example shows that the feasible price set can significantly deviate from the optimal uniform price.
\begin{example} \label{ex:optimal-price-deviate}
    Suppose buyers have three possible values: $1$, $2$, and $100$, with population distribution $(100 / 152, 50 / 152, 2 / 152)$. The seller’s optimal uniform price is $p^* = 100$, yielding a revenue of $200/152$. However, $F = \{1,2\}$ is a feasible price set. The intermediary can split the market into two segments: $(100 / 152, 0, 1 / 152)$ and $(0, 50 / 152, 1 / 152)$.The optimal price is $1$ in the first segment and $2$ in the second. Therefore, $F = \{1,2\}$ is a feasible price set, even though $p^* = 100$ is excluded from $F$. This example illustrates that feasibility is not strictly constrained by proximity to the uniform price, underscoring the broader applicability of our results.
\end{example}

\paragraph{Sufficient Condition.}
We further provide sufficient conditions for sets without optimal uniform price being feasible and then use a simulation with a uniform demand distribution to verify the prevalence of such sets.


\begin{theorem} \label{thm:uni-feasible-sufficient} (Proof see \cref{app:proof_sufficient}) 
    Suppose $v_{i^*}$ is the unique optimal price of $\bx^*$, \textit{i.e.}, $\optPrice(\bx^*) = \{v_{i^*}\}$. Consider a regulated price set $F = \{v_{\ell}, v_{\ell+1}, \dots, v_r\}$ with $v_{i^*} \not\in F$. If \[\sum_{j = \ell}^r x_j^* \cdot v_j + v_\ell \cdot \sum_{j = r + 1}^n x_j^* \ge R_{\uni}(\bx^*),\] then $F$ is $\bx^*$-feasible.
\end{theorem}

\paragraph{Simulation.}
We further perform a simulation to assess how frequently the feasible set $F$ excludes the optimal price $\bx^*$. Specifically, we let $\bx^*$ be uniformly distributed over $V = \{L, L+1, \dots, R\}$, with $R \in \{9, 49, 99, 199\}$ and $L$ varying from $1$ to $R$. We then enumerate all regulated sets that do not contain the optimal price and compute the proportion of such sets that remain feasible. \cref{fig:feasibility} plots this proportion as a function of $L/R$. For comparison, we also apply the sufficient condition from \cref{thm:uni-feasible-sufficient} and include the resulting proportion in the same figure.

\begin{figure}[h]
    \centering
    \includegraphics[width=\linewidth]{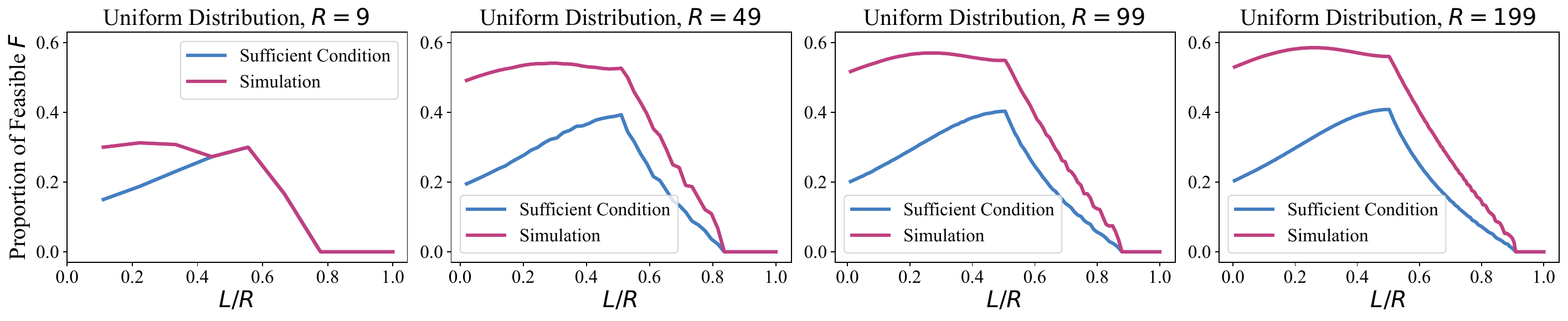}
    \caption{Proportion of feasible sets among all regulated sets that exclude the optimal price. Details are provided in \cref{sect:justification-feasibility}.}
    \label{fig:feasibility}
\end{figure}

There are two observations from \cref{fig:feasibility}. First, feasible sets excluding the optimal price are relatively common, with the proportion exceeding 0.2 in more than 60\% of cases. Second, our sufficient condition serves as a lower bound on the actual proportion, effectively capturing a significant portion of it.

\subsubsection{Proof of \cref{thm:uni-feasible-sufficient}} \label{app:proof_sufficient}
\begin{lemma} \label{lemma:proof-feasibility-optimal}
    Let $\ZPP$ denote the output and let $\bx^{\remain}$ be defined as in \cref{line:bx-remain} of \cref{alg:construction}, when \cref{eq:construction-PS-max} is used in \cref{line:construction-equation}. If $F$ is not $\bx^*$-feasible and $v_{i^*} \in \optPrice(\bx^*)$, then $v_{i^*} \in \optPrice(\bx^{\remain})$.
\end{lemma}
\begin{proof}
    According to \cref{thrm:passive-PS}, when $F$ is not $\bx^*$-feasible, we have $\bx^{\remain} > 0$.

    Assume that $v_{i^*} \notin \optPrice(\bx^{\remain})$. Let $v_j$ be any value in $\optPrice(\bx^{\remain})$. Since $\bx^{\remain} > 0$, we have $x^{\remain}_j > 0$ and $j \ne i^*$. As a result, $R_{\bx^{\remain}}(v_j) > R_{\bx^{\remain}}(v_{i^*})$. In addition, since $x^{\remain}_j > 0$, by the construction of \cref{alg:construction,eq:construction-PS-max}, $v_j \in \support(\bx^{\equal})$ for any constructed $\bx^{\equal}$ in \cref{alg:construction}. Hence, for any $(\bx, p) \in \ZPP$, we must have $v_j \in \optPrice(\bx)$ and hence $R_{\bx}(v_j) \ge R_{\bx}(v_{i^*})$. Therefore,
    $$
    R_{\bx^*}(v_j) = R_{\bx^{\remain}}(v_j) + \sum_{(\bx, p) \in \ZPP}R_{\bx}(v_j) > R_{\bx^{\remain}}(v_{i^*}) + \sum_{(\bx, p) \in \ZPP}R_{\bx}(v_{i^*}) = R_{\bx^*}(v_{i^*}),
    $$
    which leads to a contradiction. Now the claim follows.
\end{proof}

\begin{proof}[Proof of \cref{thm:uni-feasible-sufficient}]
    Suppose that $F$ is not $\bx^*$-feasible. By \cref{lemma:proof-feasibility-optimal}, after running \cref{alg:construction} on $\bx^*$ using \cref{eq:construction-PS-max} in \cref{line:construction-equation}, we obtain $v_{i^*} \in \optPrice(\bx^{\remain})$ and $x_{i^*}^{\remain} > 0$. Therefore,
    \[
    v_{i^*} \cdot \sum_{j = i^*}^n x^\remain_{j} > v_{\ell} \cdot \sum_{j = r + 1}^n x^\remain_{j}.
    \]

    By the construction of \cref{alg:construction}, $\ZPP$ can be written as $\ZPP = \{(\bx_q, p_q)\}_{q \in [|F|]}$ with $p_q = v_{\ell + q - 1}$ and $\support(\bx_q) \cap F = v_{\ell + q - 1}$. Since $\ZPP$ consists of equal-revenue markets and $v_{i^*}$ is in the support of each market, we have 
    \[
        \forall q\in [|F|], v_{i^*} \cdot \sum_{j = i^*} ^ n x_{q, j} \ge v_{\ell + q - 1} \cdot \left(x_{q, \ell + q - 1} + \sum_{j = r + 1} ^ n x_{\ell, q} \right).
    \]
    Summing up the above two inequalities after enumerating $q$, we have 
    \begin{align*}
        v_{i^*} \cdot \sum_{j = i^*} ^ n x^\remain_{j} + v_{i^*} \cdot \sum_{q = 1}^{|F|}  \sum_{j = i^*} ^ n x_{q, j} &\ge \sum_{q = 1} ^ {|F|} v_{\ell + q - 1} \cdot \left(x_{q, \ell + q - 1} + \sum_{j = r + 1} ^ n x_{j, q} \right) + v_{\ell} \cdot \sum_{j = r + 1} ^ n x^\remain_{j}
        \\ 
        & \ge \sum_{q = 1} ^ {|F|} v_{\ell + q - 1} \cdot x_{q, \ell + q - 1} + \sum_{q = 1}^{|F|} v_{\ell} \cdot \sum_{j = r + 1}^n x_{j, q} + v_{\ell} \cdot \sum_{j = r + 1} ^ n x^\remain_{j}
        \\
        & = \sum_{q = 1} ^ {|F|} v_{\ell + q - 1} \cdot x_{q, \ell + q - 1} + v_{\ell} \cdot \left(\sum_{q = 1}^{|F|} \sum_{j = r + 1}^n x_{j, q} +  \sum_{j = r + 1} ^ n x^\remain_{j} \right).
    \end{align*}
    Since we have $\bx^{\remain} + \sum_{q = 1}^{|F|}\bx_q = \bx^*$, the left hand side is equal to $v_{i^*} \cdot \sum_{j = i^*}^n x^*_j = R_{\uni}(\bx^*)$, and the right hand side is equal to $\sum_{j=\ell} ^ {r} v_j \cdot x^*_j + v_{\ell} \cdot \sum_{r + 1}^n x_j^*$. This leads to a contradiction, completing the proof.
\end{proof}

\subsection{Discussion on Applying Bergemann et al. [2015]'s Method} \label{app:bbm}
\cref{ex:intro_four,ex:intro_four_continue} show that when $F$ excludes all optimal uniform prices, applying \citet{bergemann2015limits}'s method does not guarantee a feasible market segmentation. We further examine cases where $F$ includes an optimal uniform price. While the method succeeds in the consumer surplus-maximizing scheme, it fails to capture the full producer-consumer surplus region.

\paragraph{Effectiveness in the Consumer Surplus Maximizing \SCHEME.}  
We show that if $F$ contains an optimal uniform price $p^*$, the segmentation method of \citet*{bergemann2015limits} can be applied to achieve \cref{thm:informal-main}.

\begin{theorem}
    When $F = \{v_{\ell}, v_{\ell+1}, \dots, v_r\}$ intersects with the optimal uniform price set, i.e., $F \cap \optPrice(\bx^*) \ne \emptyset$, \cref{thm:informal-main} follows directly from applying \citet{bergemann2015limits}'s method.
\end{theorem}

\begin{proof}
    Let $\X = \{\bx_q\}_{q \in [Q]}$ be the output of \citet{bergemann2015limits}'s method (\cref{eq:construction-bbm}). Since $F \cap \optPrice(\bx^*) \ne \emptyset$, suppose $p^* \in F \cap \optPrice(\bx^*)$. By construction, for any $\bx_q \in \X$, we have $\optPrice(\bx^*) \subseteq \optPrice(\bx_q) = \support(\bx_q)$. Consequently, it follows that
    $$
    \support(\bx_q) \cap F = \optPrice(\bx_q) \cap F \supseteq \optPrice(\bx^*) \cap F \ne \emptyset.
    $$
    For any $q \in [Q]$, choose $p_q$ as $p_q = \min(\support(\bx_q) \cap F)$. As a result, it holds that
    $$
    \begin{aligned}
        \cs(\bx_q, p_q) & = \sum_{i=1}^n \bbI[v_i \ge p_q](v_i-p_q)x_{q,i} = \sum_{i=\ell}^n (v_i-p_q)x_{q,i} \\
        & = \sum_{i=\ell}^n v_ix_{q,i} - R_{\bx_q}(p_q) = \sum_{i=\ell}^n v_ix_{q,i} - R_{\bx_q}(p^*) \\
        \ps(\bx_q, p_q) & = R_{\bx_q}(p_q) = R_{\bx_q}(p^*).
    \end{aligned}
    $$
    Let $\Z = \{(\bx_q, p_q)\}_{q \in [Q]}$. As a result, it holds that
    $$
    \begin{aligned}
        \CS(\Z) & = \sum_{q=1}^Q\cs(\bx_q, p_q) = \sum_{q=1}^Q\sum_{i=\ell}^n \left(v_ix_{q,i} - R_{\bx_q}(p^*)\right) \\
        & = \sum_{i=\ell}^n v_i\cdot x_i^* - R_{\bx^*}(p^*) = \sum_{i=\ell}^n v_i\cdot x_i^* - \UniRev(\bx^*) \\
        \PS(\Z) & = \sum_{q=1}^Q\ps(\bx_q, p_q) = \sum_{q=1}^QR_{\bx_q}(p^*) = R_{\bx^*}(p^*) = \UniRev(\bx^*).
    \end{aligned}
    $$
    Now the claim follows.
\end{proof}

\paragraph{Limitation in Capturing the Full Consumer-Producer Surplus Region.}  
However, even if $F$ contains an optimal uniform price $p^*$, directly applying \citet{bergemann2015limits}'s method does not fully capture the entire producer-consumer surplus region.

Specifically, consider the social welfare minimizing \scheme, where a direct application of \citet{bergemann2015limits}'s method proceeds as follows. Let $\X = \{\bx_q\}_{q=1}^Q$ be the output of the market segmentation by \cref{eq:construction-bbm}, and define $p_q = \max(\support(\bx_q) \cap F)$ (noting that $\support(\bx_q) \cap F$ is non-empty by the previous proof). However, the following example demonstrates that the scheme $\Z = \{(\bx_q, p_q)\}_{q=1}^Q$ may fail to achieve the minimal possible social welfare.

\begin{example}
    Suppose $V = \{1, 2, 3, 4\}$ and $\bx^* = (0.64, 0.08, 0.04, 0.24)$. Let $F = \{1, 2, 3\}$. The unique optimal uniform price is $p^* = 1 \in F$. Applying \citet{bergemann2015limits}'s method could obtain the following market segments $\bx_1 = (0.24, 0.08, 0.04, 0.12)$, $\bx_2 = (0.36, 0, 0, 0.12)$, and $\bx_3 = (0.04, 0, 0, 0)$. Using the above method, it holds that $p_1 = 3$, $p_2 = 1$, and $p_3 = 1$. The corresponding scheme $\Z = \{(\bx_q, p_q)\}_{q=1}^3$ has social welfare
    $$
    \SW(\Z) = \sum_{q=1}^3\sw(\bx_q, p_q) = 0.6 + 0.84 + 0.04 = 1.48.
    $$
    However, consider the alternative \scheme $\Z' = \{(\bx'_q, p'_q)\}_{q = 1}^3$, where $\bx'_1 = (0.32, 0, 0.04, 0.12)$, $\bx'_2 = (0.16, 0.08, 0, 0.08)$, and $\bx'_3 = (0.16, 0, 0, 0.04)$, with corresponding prices $p'_1 = 3$, $p'_2 = 2$, and $p'_3 = 1$. It follows that
    $$
    \SW(\Z') = \sum_{q=1}^3\sw(\bx'_q, p'_q) = 0.6 + 0.48 + 0.32 = 1.4 < \SW(\Z).
    $$
    As a result, applying \citet{bergemann2015limits}'s method fail to find the social welfare minimizing \scheme.
\end{example}
\subsection{Maximizing the Minimum Consumer Surplus from the Regulator's Perspective} \label{app:final_discussion}
When the regulator aims to set $F$ given the market information, a natural question is which choice of $F$ brings the largest consumer surplus guarantee, \emph{regardless of} how the segmentation is designed by the intermediary. Observing that the left boundary of the CS-PS feasible region is a vertical line, the question is equivalent to find $F$ that maximizes $\PassiveCSMin(\bx^*, F)$. Based on the characterations in \cref{thrm:passive-SW-min}, the following corollary addresses this question. 

\begin{corollary} \label{coro:f-minimal-consuer-surplus}
    Given the market $\bx^*$ with the value set $V$. The regulation set $F$ that maximizes $\PassiveCSMin(\bx^*, F)$ is $\{v_1, v_2, \dots, v_\omega\}$, where $v_{\omega}$ is the smallest index such that $\{v_1, v_2, \dots, v_\omega\}$ is $x^*$-feasible. 
\end{corollary}
\begin{proof}
    By \cref{thrm:passive-SW-min}, $\PassiveCSMin(\bx^*, F) = \eta_0 v_{i_0} + \sum_{j=i_0+1}^n x^*_jv_j - \UniRev(\bx^*)$, where $i_0$ and $\eta_0$ are defined by \cref{eq:passive-SW-min-i0-eta0}. Since $\UniRev(\bx^*)$ is fixed, we want to maximize $\eta_0 v_{i_0} + \sum_{j=i_0+1}^n x^*_jv_j$, which is equivalent to select the largest minimal sub-feasible price set.     It can be verified that $\PassiveCSMin(\bx^*, F)$ only depends on $v_r$ as long as $F$ is $\bx^*$-feasible. Recall that
    \begin{align*}
        i_0 & = \max\{i \in \{\ell, \ell+1, \dots, r\}: \bx^* \text{ is } \{v_{i_0}, v_{i_0+1}, \dots, v_r\}\text{-feasible}\}. \\
        \eta_0 & =  \inf\{0 \le \eta \le x^*_{i_0}: (F, i_0, \eta) \text{ is } \bx^*\text{-sub-feasible}\}.
    \end{align*}
    Imagine a process where we fix $v_\ell$ in $F$ and move $v_r$ from $v_n$ to $v_1$. When $v_r$ becomes smaller, $i_0$ also becomes smaller. Moreover, by the definition of $\eta_0$, if $(F, i_0, \eta_0)$ is $\bx^*$-sub-feasible and $\eta_0 < \eta_0'$, then $(F, i_0, \eta_0')$ is also $\bx^*$-sub-feasible. Therefore, $\eta_0$ will only monotonically increase when $i_0$ is fixed. Therefore, we have the value $\eta_0 v_{i_0} + \sum_{j=i_0+1}^n x^*_jv_j$ is monotonically increasing as we $v_r$ becomes smaller. 
    
    Note that if $F$ is $\bx^*$-feasible and $F \subseteq F' \subseteq V$, then $F'$ must be $\bx^*$-feasible. Therefore, the smallest right end value $v_r$ while $F$ being $\bx^*$-feasible is taken when the left end $v_\ell$ is the smallest value in $V$. Consequently, if we find the smallest $v_\omega$ such that $\{v_1, v_2, \dots, v_\omega\}$ is feasible, $v_\omega$ is the smallest right-end value among all $\bx^*$-feasible price sets. Combining with the argument in the last paragraph completes the proof of the optimality of setting $F = \{v_1, v_2, \dots, v_\omega\}$.
\end{proof}

We present an example demonstrating that $p^*$ may not be included in the optimal (with the largest consumer surplus guarantee) regulated price set $F$ derived from \cref{coro:f-minimal-consuer-surplus}.

\begin{example}
    Let $n = 99$ and consider a uniform demand distribution over $V = \{v_1, v_2, \dots, v_n\} = \{1, 2, \dots, 99\}$, resulting in $x^*_i = 1 / 99$ for all $i$. The optimal uniform price is $p^* = 50$. Through simulation, we verify that the regulated price set $F$ obtained via \cref{coro:f-minimal-consuer-surplus} is $F = \{1, 2, \dots, 33\}$, with $p^* \not\in F$.
\end{example}
\subsection{Alternative Enumeration Methods in \cref{sect:passive-ps-max}} \label{sect:justification}
\begin{example}
    Let $V = \{1, 2, 3, 4\}$, $\bx^* = (9, 1, 1, 3)$, and $F = \{2, 3\}$. In our paper, we first calculate the equal-revenue markets supported on $\{1, 3, 4\}$ and then on $\{1, 2, 4\}$. The result is two market segments $\bx_2 = (8, 0, 1, 3)$ and $\bx_1 = (1, 1, 0, 0)$. It is easy to verify that this is feasible. However, if we first construct an equal-revenue market supported on $\{1, 2, 4\}$, the market segment will be $(2, 1, 0, 1)$ and the remaining market is $(7, 0, 1, 3)$, whose optimal price is $1 \not\in F$.
\end{example}

We can demonstrate that in the above example, the \emph{only} feasible enumeration method to achieve market segmentation in the standard form (defined in \cref{defn:standard-form}) is precisely the outcome of our method.

\begin{proposition}
    Let $V = \{1, 2, 3, 4\}$, $\bx^* = (9, 1, 1, 3)$, and $F = \{2, 3\}$. The only standard-form and \Fvalid \scheme of $\bx^*$ is $\Z = \{(\bx_1, p_1), (\bx_2, p_2)\}$ with
    $$
    \bx_1 = (1, 1, 0, 0), \quad p_1 = 2, \quad \bx_2 = (8, 0, 1, 3), \quad p_2 = 3.
    $$
\end{proposition}
\begin{proof}
    Suppose the standard-form market scheme is $\Z = \{(\bx_1, p_1), (\bx_2, p_2)\}$ with $p_1 = 2$ and $p_2 = 3$. Assume that $\bx_1 = (a, b, c, d)$ and then we have $\bx_2 = (9 - a, 1 - b, 1 - c, 3 - d)$. We have
    $$
        0 \le a \le 9, 0 \le b \le 1, 0 \le c \le 1, 0 \le d \le 3.
    $$
    Since $p_1 = 2 \in \optPrice(\bx_1)$ and $p_2 = 3 \in \optPrice(\bx_2)$, we have
    \begin{align}
        2(b+c+d) = R_{\bx_1}(2) \ge R_{\bx_1}(1) = a + b + c + d & \Rightarrow b + c + d \ge a. \label{eq:opt-1} \\
        3(4 - c - d) = R_{\bx_2}(3) \ge R_{\bx_2}(1) = 14 - a - b - c - d & \Rightarrow a + b \ge 2c + 2d + 2. \label{eq:opt-2}
    \end{align}
    \cref{eq:opt-1} $+$ \cref{eq:opt-2}, we have
    $$
        b+c+d + a + b \ge a+2c+2d+2 \Rightarrow 2b \ge c + d + 2.
    $$
    Since $0 \le b \le 1$ and $c, d \ge 0$, we have
    $$
    2 \ge 2b \ge c + d + 2 \ge 2.
    $$
    As a result, we have $b = 1$, $c = 0$, and $d = 0$. Now \cref{eq:opt-1,eq:opt-2} become $1 \ge a$ and $a \ge 1$. As a result, $a = 1$ and the claim follows.
\end{proof}

\end{document}